\documentclass[reprint,superscriptaddress,
amsmath,amssymb,
aps,
pra
]{revtex4-2}

\usepackage{qip}
\usepackage{graphicx}
\usepackage[colorlinks=true,citecolor=red,urlcolor=blue,linkcolor=blue]{hyperref}
\usepackage{amsfonts}
\usepackage{amsmath,amssymb}
\usepackage{dsfont}
\usepackage{euscript}
\usepackage{float}
\usepackage{amsthm}
\usepackage{mathalfa}
\usepackage{bbm}
\usepackage{ragged2e}
\usepackage{color}
\usepackage{xcolor}
\usepackage{pgf}
\usepackage{color,soul}
\usepackage{tikz}
\usetikzlibrary{backgrounds}
\usepackage{float}
\usepackage{multirow}
\usepackage{array}
\usepackage{mathtools}
\usepackage[caption=false, justification=centerlast]{subfig}
\usepackage[utf8]{inputenc}
\usepackage[print-unity-mantissa=false]{siunitx}

\usepackage{array}
\usepackage{booktabs}
\usepackage[caption=false]{subfig}

\usetikzlibrary{backgrounds}

\newcommand{\Tr}{\mathrm{ Tr }}
\newcommand{\norm}[1]{\left\lVert#1\right\rVert}
\newcommand{\ceil}[1]{\left \lceil #1 \right \rceil}
\newcommand{\floor}[1]{\left \lfloor #1 \right\rfloor}
\newcommand{\abs}[1]{\left\vert #1 \right \vert}

\newcommand{\dd}{\mathrm{d}}

\newcommand{\Max}{\mathrm{Max}}
\newcommand{\Min}{\mathrm{Min}}
\newcommand{\Var}{\mathrm{Var}}

\newcommand{\nmax}{n_{\mathrm{max}}}
\newcommand{\wmax}{\tau_{\mathrm{max}}}
\newcommand{\wmin}{\tau_{\mathrm{min}}}

\newcommand{\tol}{\mathrm{tol}}
\newcommand{\freq}{\mathrm{freq}}
\newcommand{\corr}{\mathrm{corr}}
\newcommand{\pass}{\mathtt{pass}}

\newcommand{\esec}{\varepsilon_\mathrm{sec}}
\newcommand{\ecor}{\varepsilon_\mathrm{cor}}
\newcommand{\ecom}{\varepsilon_\mathrm{com}}
\newcommand{\esmooth}{\varepsilon_s}
\newcommand{\eEAT}{\varepsilon_\mathrm{EA}}

\newcommand{\vb}[1]{\mathbf{#1}}
\newcommand{\vbf}[1]{\boldsymbol{#1}}

\newtheorem{thm}{Theorem}
\newtheorem{df}{Definition}
\newtheorem{lemma}{Lemma}
\newtheorem{corollary}{Corollary}

\newcounter{protcounter}

\begin{document}
\title{Discrete-modulated continuous-variable quantum key distribution \\
secure against general attacks}

\author{Ignatius William Primaatmaja}
\email{williamp@squareroot8.com}
\affiliation{Department of Electrical \& Computer Engineering, National University of Singapore, Singapore}
\affiliation{Squareroot8 Technologies Pte Ltd, Singapore}

\author{Wen Yu Kon}
\affiliation{Department of Electrical \& Computer Engineering, National University of Singapore, Singapore}
\affiliation{Global Technology Applied Research, JPMorgan Chase \& Co, Singapore}

\author{Charles Lim}
\affiliation{Department of Electrical \& Computer Engineering, National University of Singapore, Singapore}
\affiliation{Global Technology Applied Research, JPMorgan Chase \& Co, Singapore}

\begin{abstract}
In recent years, discrete-modulated continuous-variable quantum key distribution (DM-CV-QKD) has gained traction due to its practical advantages: cost-effectiveness, simple state preparation, and compatibility with existing communication technologies. This work presents a security analysis of DM-CV-QKD against general sequential attacks, including finite-size effects. Remarkably, our proof considers attacks that are neither independent nor identical, and makes no assumptions about the Hilbert space dimension of the receiver. To analyse the security, we leverage the recent generalised entropy accumulation theorem and the numerical methods based on quasi-relative entropy. We also develop a novel dimension reduction technique which is compatible with the entropy accumulation framework. While our analysis reveals significant finite-size corrections to the key rate, the protocol might still offer advantages in specific scenarios due to its practical merits. Our work also offers some insights on how future security proofs can improve the security bounds derived in this work.
\end{abstract}

\maketitle
\section{Introduction} \label{sec: introduction}
Quantum key distribution (QKD) is amongst the most mature quantum technologies, with some companies pushing for its commercial applications. Since its inception in 1984~\cite{Bennett2014quantum}, QKD has undergone a rapid development in both theoretical and experimental aspects. Novel QKD protocols have been proposed and demonstrated, typically with an
emphasis on improving the secret key rate and/or the range of the protocol. 
However, for QKD to be widely deployed, it is also crucial to take into account the
simplicity of the setup as well as the cost of the hardware.

Broadly speaking, based on how classical information is being encoded into quantum information, QKD protocols can be distinguished into two major categories: discrete-variable (DV) and continuous-variable (CV). DV-QKD protocols encode discrete classical information into optical modes (e.g., polarisation, time-bins, etc) of a light source, typically a laser. To decode the classical information, photon counting techniques are used, usually using threshold single-photon detectors. DV-QKD typically has good secret key rates and range (e.g., ~\cite{frohlich2017long, boaron2018secure, li2023high}), but the single-photon detectors required are less mature, more expensive and sometimes require cooling.

On the other hand, CV-QKD protocols typically encode continuous classical information into the optical phase space~\cite{grosshans2002continuous, diamanti2015distributing, laudenbach2018continuous,zhang2024continuous} (i.e., the $X$ and $P$ quadratures). To decode the information, coherent detection techniques, such as homodyne and heterodyne detection, are typically employed. Remarkably, these detectors are also used in the classical communication infrastructure, thus homodyne and heterodyne detectors are commercially more mature than single-photon detectors. Furthermore, they are also typically more cost-effective than the single-photon detectors counterparts and they can operate at room-temperature. Moreover, these detectors are readily integrated on chips~\cite{zhang2019integrated, hajomer2023continuous}. Therefore, the use of coherent detection techniques provide an edge for practical applications, especially over metropolitan distances and last-mile key exchanges~\cite{zhang2019continuous,zhang2020continuous, wang2022sub, roumestan2022experimental, wang2024high, williams2024field}. 

However, the requirement that classical information has to be continuous (often Gaussian distributed~\cite{grosshans2002continuous}) affects the practicality of CV-QKD protocols. Real optical modulators have finite precision which means that continuous modulation is only a theoretical idealisation, which does not correspond to the practical implementation. Such theory-experiment gap can open up security loopholes for QKD and it may be exploited by a sufficiently advanced adversary. To ease this requirement, two solutions has been proposed 
in the literature. Firstly, one can encode discrete information into optical modes, like in DV-QKD, and use coherent detection techniques to perform photon counting~\cite{qi2021bb84, primaatmaja2022discrete, jin2023pilot}. This method has the additional benefit of removing the requirement of sharing a global phase reference between the two communicating parties at the price of shorter distances. Another alternative is to simply encode discrete classical information into a discrete constellation in the phase space. This approach is known as discrete-modulated CV-QKD (DM-CV-QKD). While this solution requires the two honest parties to share a global phase reference, protocols with this feature typically have longer ranges and higher secret key rate. In this work, we shall focus our attention to the latter.

For practical implementations, the quantum communication phase in QKD is performed round-by-round. Attacks on QKD can be classified based on what the adversary can do in the quantum channel in each round~\cite{scarani2009security}. For the most general attack (also known as coherent attacks), the adversary is allowed to execute any strategy that is allowed by quantum mechanics (e.g., adjusting her attack based on a measurement result in the preceding rounds, entangle multiple states together before sending them to the receiver, etc). Analysing the security of a QKD protocol against general attack may require multi-round analyses, would be intractable to parameterise as QKD protocols typically involve a large number of rounds.

Nevertheless, techniques have been developed for protocols that exhibit certain structures (e.g., permutation symmetry~\cite{renner2007symmetry, renner2009de_finetti, christandl2009postselection, leverrier2017security}, Markov condition in sequential protocols~\cite{dupuis2019entropy, dupuis2020entropy}, non-signalling condition in sequential protocols~\cite{metger2022generalised, metger2023security}, etc). In these special cases, the techniques reduce the problem to single-round analyses. Interestingly, many of these techniques use the conditional von Neumann entropy of the secret bit, given the adversary's quantum side information and the classical announcements for a single round, as the quantity of interest. This is the same quantity of interest for security under collective attacks, where an adversary acts independently and identically in each round, allowing the security of the protocol to be analyzed by only looking at a single round of the protocol. This causes many of the works that are tailored to analysing the security against collective attacks to be relevant to the security analysis against general attacks.

One of these techniques is the \textit{generalised entropy accumulation theorem} (GEAT)~\cite{metger2022generalised, metger2023security}. It is a technique that is applicable to any QKD protocols that has the so-called \textit{no-signalling} structure. Informally, the no-signalling structure requires that classical and quantum information leaked in a given round does not leak additional information about the secret bits generated in preceding rounds. Fortunately, many QKD protocols exhibit this feature which makes the GEAT a versatile tool to prove the security of a QKD protocol. Another important feature of the GEAT is that it naturally incorporates the finite-size effects into the security analysis. As any practical QKD protocol is executed within a finite duration, finite-size effects such as statistical fluctuations, have to be accounted for carefully.

DM-CV-QKD protocols have enjoyed a significant progress in their security analysis in recent years. First, the security of the protocol is analysed under restricted attacks~\cite{leverrier2009unconditional, heid2007security}. Then, Ghorai et al~\cite{ghorai2019asymptotic} has proven its security in the asymptotic limit (i.e., when there are infinite number of rounds). The security proof relies on the extremality of Gaussian states, a technique that is commonly used in Gaussian modulated CV-QKD. However, for DM-CV-QKD, it turns out that the bound is too pessimistic, which leads to sub-optimal secret key rate. The asymptotic bound was then improved by Lin et al~\cite{lin2019asymptotic} under the so-called \textit{photon number cutoff assumption}. This assumption effectively assumes that the adversary sends a state that lives in a finite-dimensional Hilbert space (corresponding to the low photon number subspace of the Fock space). While this assumption may be a good approximation which greatly simplifies the computation, it is hardly justifiable. However, the work of Lin et al~\cite{lin2019asymptotic} showed that DM-CV-QKD could potentially allow the exchange of secret keys at relatively high key rate and over reasonably long distances.

Following their initial work in Ref.~\cite{lin2019asymptotic}, Lin-L\"utkenhaus also analysed the security of the DM-CV-QKD protocol under the trusted detector noise scenario~\cite{lin2020trusted}. This refers to the scenario where the limitations of the detectors, such as imperfect quantum efficiency and electrical noise, are characterised and known to the users. In the normal circumstances, these imperfections are attributed to the adversary. Ref.~\cite{lin2020trusted} showed that by characterising these imperfections, the performance of the protocol can be greatly improved. Another important milestone is presented in Ref.~\cite{upadhyaya2021dimension}, which removes the need for photon number cutoff assumption. The work of Ref.~\cite{upadhyaya2021dimension} offered the so-called \textit{dimension reduction technique}, which effectively gives a correction term to the computation of the conditional von Neumann entropy with the photon number cutoff. Remarkably, Ref.~\cite{upadhyaya2021dimension} showed that the correction term is relatively small and DM-CV-QKD can still offer practical advantages even after accounting for this correction.

Given the potential of DM-CV-QKD shown from the studies of the protocol's security in the asymptotic regime, in recent years, researchers have tried to study the finite-size security of the protocol. Notably, Kanitschar et al~\cite{kanitschar2023finite} studied the finite-size security of the protocol under the assumption that the adversary performs collective attacks. On the other hand, Bauml et al~\cite{bauml2023security} studied the finite-size security of the protocol against general attacks but under the photon number cutoff assumption.

Remarkably, while these works focused on protocols with quadrature phase-shift keying (QPSK) encoding, DM-CV-QKD protocols with different constellations have also been studied. The two-state protocol has been analysed by Ref.~\cite{zhao2009asymptotic} in the asymptotic limit and by Ref.~\cite{matsuura2021finite, matsuura2023refined} in the finite-key regime. On the other hand, the three-state protocol was analysed in Ref.~\cite{bradler2018security}. Recently, the variant of the protocol with large constellation size has also been analysed~\cite{kaur2021asymptotic,denys2021explicit}.

At the time of writing, the security of DM-CV-QKD protocols against general attacks have not been performed without the photon number cutoff assumptions (except for the two-state case~\cite{matsuura2021finite, matsuura2023refined}). Unfortunately, it is not straightforward to remove the collective attack assumption using Ref.~\cite{kanitschar2023finite}'s framework since techniques that relies on permutation symmetry has a correction term that depends on the Hilbert space dimension, which is infinite for DM-CV-QKD.
Moreover, entropy accumulation frameworks require the parameter estimation to be based on frequency distributions while Ref.~\cite{kanitschar2023finite} estimates moments. On the other hand, it is also difficult to remove the photon number cutoff assumption in Ref.~\cite{bauml2023security} as the dimension reduction technique requires the estimate of the mean photon number of the incoming states, which is again incompatible with the entropy accumulation framework. 

In this work, we successfully remove these two assumptions simultaneously by employing the entropy accumulation framework (in particular, the GEAT version~\cite{metger2022generalised, metger2023security}) and developing a novel dimension reduction technique that only requires estimation of probability distributions. However, our work reveals that when both the collective attack and photon number cutoff assumptions are removed, DM-CV-QKD protocols suffer from a significant finite-size penalty.

The paper is organised as follows. We first introduce some preliminary notions in Section~\ref{sec: prelim}, such as the notation that we use in this paper, how security is defined for QKD and the GEAT framework which will be used in this paper. In Section~\ref{sec: protocol}, we define the DM-CV-QKD protocol that we will analyse. We will present the security analysis of the protocol in Section~\ref{sec: security}. Following which, we present a simulation of the protocol's performance in Section~\ref{sec: simulation}. Lastly, we will discuss the results that we obtain, some potential future works, and we conclude our paper in Section~\ref{sec: conclusion}.

\section{Preliminaries} \label{sec: prelim}
\subsection{Notations and definitions} \label{subsec: notation}
\subsubsection{Basic notations}

We will first introduce the notations that we use in this paper. The capital letters $A$, $B$, and $E$ are typically associated to the parties Alice, Bob and Eve, respectively. 
The other capital letters from the Latin alphabet typically refer to random variables, unless explicitly stated. In particular, we denote $N$ as the total number of rounds in the QKD protocol. Many of the random variables in this work are indexed in the subscripts to indicate the round associated to the particular random variable. For example, we use the letter $Z$ to denote Bob's measurement outcome. Then, $Z_j$ denotes Bob's measurement outcome in the $j$-th round. On the other hand, when we index a random variable in the superscript, this refers to a sequence of random variables up to that particular round. For example, the random variable $Z^j$ refers to the sequence $Z_1, Z_2 , ..., Z_j$. When referring to an entire sequence, we will write the particular letter in boldface. For example, when writing $\vb{Z}$, we refer to the sequence $Z_1, Z_2 ,..., Z_N$. When referring to a key of length $\ell$, we will also use the boldface letter $\vb{K} = K_1, K_2, ... , K_\ell$ where $K_j$ refers to the $j$-th bit of the key. For any positive integer $j$, we write $[j]$ to denote the set $\{1, 2, ..., j\}$.

The calligraphic capital letters (e.g., $\cN$, $\cM$, $\cT$, etc) are typically reserved for quantum channels. There are certain exceptions for this rule, which we will explicitly state when we first introduce the notation in the relevant section. Notable exceptions for this rule is $\cC$ which denotes the range of the random variable $C$, and $\cH$ which we generally reserve to denote Hilbert spaces.

We use $\mathbb{N}, \mathbb{Z}$ and $\mathbb{R}$ to denote the sets of natural numbers, integers and real numbers, respectively. We write $\mathsf{L}(\cH)$ to denote the set of linear operators acting on the Hilbert space $\cH$, $\mathsf{B}(\cH)$ denotes the set of bounded operators acting on the Hilbert space $\cH$, $\mathsf{D}(\cH)$ denotes the set of sub-normalised quantum states living in the Hilbert space $\cH$.

\subsubsection{Distance measures}
In QKD, we often consider how close a quantum state to another quantum state. This is made precise using distance measures. In this paper, there are two distance measures that we use. Firstly, we use $\Delta$ to denote the trace distance, which is defined as
\begin{equation}
    \Delta(\rho, \sigma) = \frac{1}{2} \norm{\rho - \sigma}_1,
\end{equation}
for some $\rho, \sigma \in \mathsf{D}(\cH)$ and some Hilbert space $\cH$. For some operator $M$, the trace norm is defined as
\begin{equation}
    \norm{M}_1 = \Tr[\sqrt{M^\dagger M}],
\end{equation}
where $M^\dagger$ denotes the adjoint of $M$. The trace distance is used in defining the security of QKD.

The second distance measure that we use is the purified distance, denoted by $\Delta_P$. It is defined as
\begin{equation}
    \Delta_P(\rho,\sigma) = \sqrt{1 - F(\rho, \sigma)^2},
\end{equation}
for some $\rho, \sigma \in \mathsf{D}(\cH)$ and some Hilbert space $\cH$. The generalised fidelity $F(\rho,\sigma)$ is defined as\
\begin{equation}
    F(\rho,\sigma) = \norm{\sqrt{\rho} \sqrt{\sigma}}_1 + \sqrt{(1 - \Tr[\rho])(1-\Tr[\sigma])}
\end{equation}
In this work, the purified distance is used to define the conditional smooth min-entropy.

\subsubsection{Entropic quantities} 
The quantum communication phase of a QKD protocol can be thought of as a weak random source -- it generates random strings that may not be uniformly distributed nor independent of the adversary's side information. To characterise the quality of a weak random source, it is customary to consider its ``entropy''. There are many entropic quantities in the literature, each appropriate to a particular information processing task. We shall introduce some of these entropic quantities that we will use in this work.

We recall that the relevant scenario for QKD is one where there is a classical register held by the honest parties and a quantum register held by the adversary. In this section, we shall denote the classical register by $A$ and the quantum register by $B$. Let $\rho_{AB} \in \mathsf{D}(\cH_A \otimes \cH_B)$ be a classical-quantum state. We can write $\rho_{AB} = \sum_{a} p(a) \ketbra{a}{a} \otimes \rho_B^a$, for some probability distribution $\{p(a)\}_a$. We will define the entropic quantities with respect to this state.

The first entropic quantity that we consider is the conditional min-entropy, which is relevant for the task of privacy amplification. This quantity is a measure of the probability of guessing the value of the register $A$, given the quantum side information stored in register $B$. The conditional min-entropy is defined as
\begin{equation}
    H_{\min}(A|B)_{\rho} = - \log_2 \max_{\{\Pi_a\}_a} \sum_{a} p(a) \Tr[\rho_B^a \Pi_a].
\end{equation}
Here, the maximisation is taken over all possible measurement operators $\{\Pi_a\}_a$ acting on $\cH_B$.

However, when allowing the output of a QKD protocol to slightly deviate from an ideal key, we can consider \textit{smoothing} on the conditional min-entropy. In this case, given the classical-quantum state $\rho_{AB}$, we do not actually evaluate the conditional min-entropy on the state $\rho_{AB}$, but on another state that is close to $\rho_{AB}$. Now, let $\esmooth \in (0,1)$. $\esmooth$ is the so-called smoothing parameter, which measures how close the state in which the conditional min-entropy is evaluated to the state under consideration, i.e., $\rho_{AB}$. In this work, the distance measure that we use for the smoothing is the purified distance. We define the $\esmooth$-ball around the state $\rho_{AB}$, which we denote as $\cB_{\esmooth}(\rho_{AB})$
\begin{multline}
    \mathcal{B}_{\esmooth}(\rho_{AB}) = \{\sigma_{AB}:  \sigma_{AB} \in \mathsf{D}(\cH_A \otimes \cH_B), \\ \Delta_P(\rho_{AB}, \sigma_{AB}) \leq \esmooth\}.
\end{multline}
Given the smoothing parameter $\esmooth$ and state $\rho_{AB}$, the conditional smooth min-entropy is defined as
\begin{equation}
    H_{\min}^{\esmooth}(A|B)_{\rho} = \max_{\sigma_{AB} \in \cB_{\esmooth}(\rho_{AB})} H_{\min}(A|B)_{\sigma}.
\end{equation}
In general, there exists a state that is $\esmooth$-close to $\rho_{AB}$ (in purified distance) but has significantly higher conditional min-entropy than $\rho_{AB}$. This means that in general, the conditional smooth min-entropy can be significantly higher than its non-smoothed counterpart. Therefore, the conditional smooth min-entropy gives a more optimal characterisation of randomness of a weak random source

The next entropic quantity that we consider is the conditional von Neumann entropy, defined as
\begin{equation}
    H(A|B)_{\rho} = H(AB)_{\rho} - H(B)_{\rho}
\end{equation}
where $H$ denotes the von Neumann entropy, defined as
\begin{equation}
\begin{split}
    H(AB)_{\rho} &= - \Tr[\rho_{AB} \log_2(\rho_{AB})]\\
    H(B)_{\rho} &= - \Tr[\rho_B \log_2(\rho_B)].
\end{split}
\end{equation}
The conditional von Neumann entropy is a relevant quantity when using the generalised entropy accumulation theorem.

Finally, we also introduce the quantum relative entropy. Given a quantum state $\rho$ and positive semi-definite operator $\sigma$, the quantum relative entropy is given by
\begin{equation}
    D(\rho||\sigma) = 
    \begin{cases}
        \Tr[\rho (\log_2(\rho) - \log_2(\sigma))] & \text{if $\mathrm{supp}(\rho) \subseteq \mathrm{supp}(\sigma)$} \\
        +\infty & \text{otherwise},
    \end{cases}
\end{equation}
where for an operator $M \in \mathsf{L}(\cH)$, $\mathrm{supp}$ is defined as
\begin{equation}
    \mathrm{supp}(M) = \{\ket{\psi} \in \cH: M\ket{\psi} \neq 0\}
\end{equation}

Given a classical-quantum state $\rho_{AB}$, the conditional von Neumann entropy is related to the quantum relative entropy
\begin{equation}
    H(A|B)_{\rho} = D(\rho_{AB}||\1_A \otimes \rho_B).
\end{equation}

\subsection{Security definition} \label{subsec: sec def}
The goal of a QKD protocol is to allow two distant parties, whom we call Alice and Bob, to share a pair of identical keys that is random from the point of view of the adversary, whom we call Eve. In the literature, this is often formalised in terms of the ``real-versus-ideal-world'' paradigm. Under this paradigm, we consider two protocols -- the real QKD protocol, which outputs the state $\rho_{\vb{K}_A \vb{K}_B \vb{L} E}$ (where $\vb{K}_A$ stores Alice's key, $\vb{K}_B$ stores Bob's key, $\vb{L}$ stores Eve's classical side information and $E$  stores Eve's quantum side information), and a hypothetical ideal protocol, which outputs a pair of perfect keys whenever the protocol is not aborted -- we denote the event in which the protocol is not aborted by $\Omega$. In particular, this perfect pair of keys are always identical, uniformly distributed and independent from the adversary's side information $E$ and $\vb{L}$.

In practice, when proving the security of a QKD protocol, we usually consider the following criteria
\begin{enumerate}
    \item \textit{Correctness}:\\
    For a fixed $\ecor \in (0,1)$, a QKD protocol is said to be $\ecor$-correct if it satisfies 
    \begin{equation}
        \Pr[\vb{K}_A \neq \vb{K}_B \wedge \Omega] \leq \ecor.
    \end{equation}
    
    \item \textit{Completeness}:
    For a fixed $\ecom \in (0,1)$, a QKD protocol protocol is said to be $\ecom$-complete if
    \begin{equation}
        \Pr[\Omega]_{\mathrm{honest}} \geq 1 - \ecom.
    \end{equation}
    The subscript ``honest'' specifies that the probability is evaluated based on the honest implementation. This refers to the scenario where the implementation might be noisy, but the adversary does not introduce more noise than expected.
    
    \item \textit{Secrecy}:\\
    For a fixed $\esec \in (0,1)$ and $\ell \in \mathbb{N}$, a QKD protocol is said to be $\esec$-secret (with respect to Bob's key) if
    \begin{equation}
        \Pr[\Omega] \cdot \Delta \left(\rho_{\vb{K}_B \vb{L} E|\Omega}, \tau_\ell \otimes \rho_{\vb{L} E|\Omega} \right)\leq \esec,
    \end{equation}
    where $\ell$ denotes the output length of the QKD protocol and $\tau_{\ell} = 2^{-\ell}\sum_{k \in \{0,1\}^\ell} \ketbra{k}{k}$.

    Specifically for secrecy, the conditional smooth min-entropy is central in the security analysis of QKD, due to the so-called \textit{leftover hash lemma}. Let $\cM_{\mathrm{PA}}$ be the quantum channel that describes the \textit{privacy amplification} step of a QKD protocol. We write $\rho_{\vb{K}_B \vb{L} E|\Omega} = \cM_{\mathrm{PA}}[\rho_{\vb{Z} \vb{L}E|\Omega}]$. The leftover hash lemma states that if two-universal hashing is used, then we have~\cite{tomamichel2011leftover}
    \begin{multline}
        \Delta(\rho_{\vb{K}_B \vb{L} E|\Omega}, \tau_{\ell} \otimes \rho_{\vb{L} E|\Omega}) \\
        \leq 2^{-\frac{1}{2}( H_{\min}^{\esmooth}(\vb{Z}|\vb{L},E)_{\rho_{\vb{Z}\vb{L}E|\Omega}} - \ell + 2)} + 2 \esmooth.
    \end{multline}
    In other words, the distance between the state describing the output of the real QKD protocol and the output of the ideal protocol is related to the conditional smooth min-entropy of the raw key $\vb{Z}$ (i.e., the string that is obtained before the privacy amplification step). This reduces the problem of proving the secrecy of a QKD protocol to deriving a lower bound on the conditional smooth min-entropy of the raw key. In turn, this can be solved using GEAT.
\end{enumerate}

\subsection{Generalised entropy accumulation} \label{subsec: GEAT}
Generalised entropy accumulation theorem (GEAT)~\cite{metger2022generalised, metger2023security} is a technique for lower bounding the conditional smooth min-entropy $H_{\min}^{\esmooth}(\vb{Z}|\vb{L}, E_N)$ of a string $\vb{Z} = Z_1 ... Z_N$ given the classical side information $\vb{L} = L_1 ... L_N$ and quantum side-information $E_N$ that are generated from a certain class of $N$-step sequential processes. GEAT formalises this class of processes in terms of the so-called GEAT channels, which we define below~\footnote{The definition we use in this work is slightly less general than the definition of GEAT channels used in Refs.~\cite{metger2022generalised, metger2023security}. In particular, we do not consider any memory registers since we focus on the device-dependent framework in this paper}.

\begin{df}[GEAT channels]
The quantum channels $\{\cM_j\}_{j=1}^N$ where $\cM_j: E_{j-1} L^{j-1} \rightarrow Z_j C_j E_j L^j$ (registers $Z_j$, $L_j$ and $C_j$ are classical registers) are called GEAT channels if each $\cM_j$ satisfies the following two conditions
\begin{enumerate}
    \item \emph{Non-signalling condition}: There exists a quantum channel $\cR_j: E_{j-1} L^{j-1} \rightarrow E_j L^j$ such that
    \begin{equation*}
        \Tr_{Z_j C_j} \circ \cM_j = \cR_j
    \end{equation*}
    \item For all $j \in [N]$, $C_j$ has a common alphabet $\cC$ and is a deterministic function of $Z_j$ and $L_j$.
\end{enumerate}
\end{df}
In the context of QKD, the non-signalling condition requires that in any given round, Eve is capable to update her side-information without Bob's private information that has been generated in the preceding rounds. This is indeed the case for many QKD protocols, including the DM-CV-QKD protocol that we will analyse. 

Roughly speaking, the GEAT states that
\begin{equation}
    H_{\min}^{\esmooth}(\vb{Z}|\vb{L}, E) \geq N h - O(\sqrt{N}),
\end{equation}
where the constant $h$ can be obtained by analysing a single round of the protocol, rather than analysing the entire protocol directly. In fact, the constant $h$ is related to the so-called \textit{min-tradeoff function} which we define below.

\begin{df}[Min-tradeoff functions]
Let $\{\cM_j\}_{j=1}^N$ be a sequence of GEAT channels and let $\cP_{\cC}$ be the set of probability distributions over $\cC$. An affine function $f: \cP_{\cC} \rightarrow \mathbb{R}$ is called a min-tradeoff function for the GEAT channels $\{\cM_j\}_{j=1}^N$ if it satisfies
    \begin{equation*}
        f(q) \leq \inf_{\sigma \in \Sigma(q)} H(Z_j|E_j L^j\tilde{E}_{j-1})_\sigma,
    \end{equation*}
    where
    \begin{multline*}
        \Sigma(q) = \Bigg\{ \sigma_{Z_j C_j L^j E_j \tilde{E}_{j-1}} = (\cM_j \otimes \id_{\tilde{E}_{j-1}})[\omega_{L^{j-1}E_{j-1}\tilde{E}_{j-1}}]: \\ 
        \sigma_{C_j} = \sum_{c \in \cC} q(c) \ketbra{c}{c}_{C_j} \Bigg\}.
    \end{multline*}
\end{df}

The constant $h$ is the value of the min-tradeoff function evaluated at the worst probability distribution that is accepted in the parameter estimation step of the protocol. To calculate the second order terms $O(\sqrt{N})$, it is useful to define some properties of the min-tradeoff functions
\begin{equation}
\begin{split}
    \Max[f] &= \max_{q \in \cP_{\cC}} f[q],\\
    \Min[f] &= \min_{q \in \cP_{\cC}} f[q],\\
    \Min_{\Sigma}[f] &= \min_{q: \, \Sigma(q) \neq \varnothing} f[q],\\
    \Var[f] &= \max_{q: \, \Sigma(q) \neq \varnothing} \sum_{c} q(c) f(\vbf{\hat{e}}_c)^2 - \left(\sum_{c} q(c) f(\vbf{\hat{e}}_c)\right)^2.
\end{split}
\end{equation}
Having introduced these properties of the min-tradeoff function, we can now present the lower bound on the conditional smooth min-entropy
\begin{thm}[Generalised entropy accumulation theorem~\cite{metger2022generalised}] \label{thm: GEAT}
    Let $\esmooth \in (0,1), \beta \in (0, \frac{1}{2})$. Let $\rho_{E_0}$ be an initial state and $\rho_{\vb{Z} \vb{C} \vb{L} E_N} = \cM_N \circ ... \circ \cM_1[\rho_{E_0}]$ be the final state after applying the GEAT channel $N$ times. Let $f$ be a min-tradeoff function and $\Omega \subseteq \cC^N$ be an event such that $h = \min_{\vb{c} \in \Omega} f(\freq_{\vb{c}})$. Then,
    \begin{multline}
        H_{\min}^{\esmooth}(\vb{Z}|\vb{L}, E_N)_{\rho_{\vb{Z} \vb{L} E_N |\Omega}} \\
        \geq N h - \frac{\beta}{1 - \beta} \frac{\ln 2}{2} V^2- N \left(\frac{\beta}{1 - \beta}\right)^2 K_{\beta}\\
        - \frac{\log_2(1 - \sqrt{1 - \esmooth^2}) - (1+\beta) \log_2 \Pr[\Omega]}{\beta},
    \end{multline}
    where
    \begin{align*}
        V &= \log_2(2d_Z^2 + 1) + \sqrt{2 + \Var[f]}\\
        K_{\beta} &= \frac{(1 - \beta)^3}{6(1 - 2 \beta)^3 \ln 2} 2^{\frac{\beta}{1 - \beta}(2 \log_2 d_Z + \Max[f] - \Min_{\Sigma}[f])}\\
        & \quad \times \ln^3 \left(2^{2 \log_2 d_Z + \Max[f] - \Min_{\Sigma}[f]} + e^2 \right).
    \end{align*}
\end{thm}

To construct the min-tradeoff function, it is customary to decompose the GEAT channel into two parts: the testing channel and the generation channel
\begin{equation}
    \cM_j = (1-\gamma) \cM_{j}^{\mathrm{gen}} + \gamma \cM_j^{\mathrm{test}},
\end{equation}
where the $C_j$ register of the output of $\cM_j^{\mathrm{gen}}$ is always $\ketbra{\perp}{\perp}$. On the other hand, $\cM_j^{\mathrm{test}}$ never outputs $\ketbra{\perp}{\perp}$ on $C_j$. Then, consider another affine function $g: \cP_{\cC \setminus \{\perp\}} \rightarrow \mathbb{R}$ that satisfies the following property
\begin{equation} \label{eq: min-tradeoff g}
    \begin{split}
        g(q) \leq &\inf_{\sigma} H(Z_j|E_j, L^j \tilde{E}_{j-1})_{\cM_j[\sigma]}\\
        &\text{s.t.} \, \cM_j^{\mathrm{test}}[\sigma]_{C_j} = \sum_{c \neq \perp} q(c) \ketbra{c}{c}.
    \end{split}
\end{equation}
Intuitively, $g$ is a lower bound on the conditional von Neumann entropy subject to the probability distribution of the test rounds. From $g$, we can construct the min-tradeoff function from the function $g$ by using the following construction
\begin{equation}
\begin{split}
    f(\vbf{\hat{e}}_c) &= \Max[g] + \frac{g(\vbf{\hat{e}}_c) - \Max[g]}{\gamma}, \qquad  c \in \cC \setminus \{\perp\},\\
    f(\vbf{\hat{e}}_{\perp}) &= \Max[g].
\end{split}
\end{equation}
This construction will lead to the following properties of the min-tradeoff function
\begin{equation} \label{eq: min tradeoff properties}
    \begin{split}
        \Max[f] &= \Max[g],\\
        \Min[f] &= \left(1 - \frac{1}{\gamma} \right)\Max[g] + \frac{1}{\gamma}\Min[g],\\
        \Min_{\Sigma}[f] &\geq \Min[g],\\
        \Var[f] &\leq  \frac{(\Max[g] - \Min[g])^2}{\gamma}.
    \end{split}
\end{equation}
The explicit construction of the function $g$ for DM-CV-QKD will be presented in Section~\ref{sec: security}.

\section{Protocol} \label{sec: protocol}

\begin{figure*}
\includegraphics[width=0.9\textwidth]{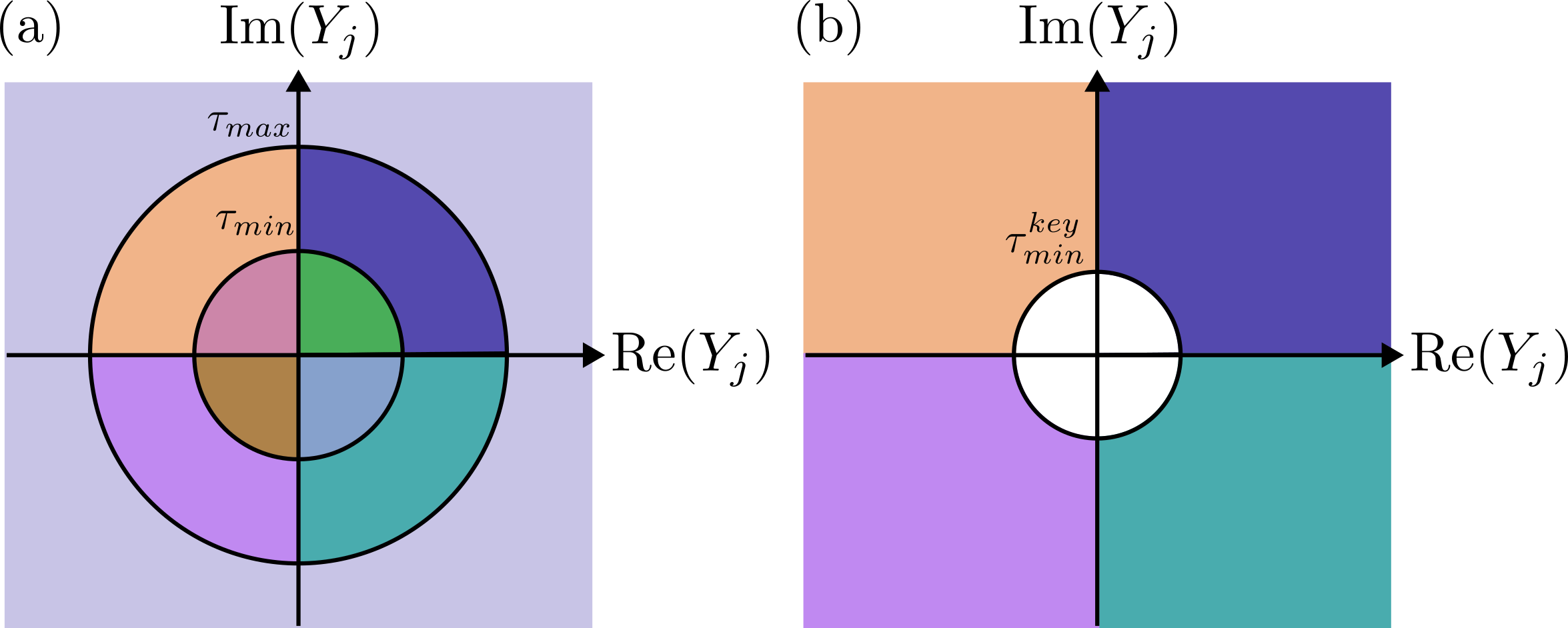}
\caption{Discretisation map for (a) $T_j=1$ and (b) $T_j=0$. When $T_j=1$, the phase space is split into 9 regions. When $T_j=0$, the phase space is split into 4 regions based on angle $\theta_j$, and excludes the central region with $W_j<\tau_{\min}^\mathrm{key}$, which is post-selected away.}
\label{fig:TestKeyRound}
\end{figure*}

We shall now give the description for the protocol that we consider. For the sake of concreteness, we consider a variant of the protocol with quadrature-phase-shift-keying (QPSK) encoding, heterodyne detection and reversed reconciliation for error-correction. However, our security analysis can be easily adapted to the case where another constellation is used for the encoding of the quantum states as we leverage on numerical methods that did not exploit any symmetry of the encoding scheme. The dimension reduction technique can also be modified to the protocol variant which uses homodyne detection rather than heterodyne detection.

We fix $N \in \mathbb{N}$, $\alpha > 0$, $\wmin^{\mathrm{key}} > 0$, $0 \leq \wmin < \wmax$, $\gamma \in (0, 1)$ and the parameter estimation scores $f_{\mathrm{PE}}$ and its tolerated range $\cF_{\tol}$.

\begin{enumerate}
    \item \textit{State preparation}: For every round $j \in [N]$, Alice uniformly chooses $X_j \in \{0,1,2,3\}$. Next, she prepares the coherent state $\ket{\psi_{X_j}} = \ket{\alpha e^{i (\pi X_j/2 + \pi/4)}}$ for some fixed $\alpha > 0$ that is specified by the protocol. She sends the state to Bob via an untrusted quantum channel.
    
    \item \textit{Measurement}: Bob randomly assigns $T_j \in \{0, 1\}$ with probability $\{1-\gamma, \gamma\}$, respectively. Upon receiving the optical signal from the untrusted quantum channel, Bob performs a heterodyne detection to obtain an outcome $Y_j \in \mathbb{C}$. 
    From the measurement outcome, he writes $Y_j = \sqrt{W_j} e^{i \theta_j}$ for some $\theta_j \in [0, 2\pi)$ and $W_j = \abs{Y_j}^2$. 
    If $T_j = 0$, he applies the following discretisation map $Y_j \rightarrow Z_j$ (see Fig.~\ref{fig:TestKeyRound}b):
    \begin{equation*}
        Z_j = \begin{cases}
        \varnothing & W_j < \wmin^{\mathrm{key}}, \\
        0 & W_j \geq \wmin^{\mathrm{key}} \wedge \theta_j \in [0 , \pi/2),\\
        1 & W_j \geq \wmin^{\mathrm{key}} \wedge \theta_j \in [\pi/2 , \pi),\\
        2 & W_j \geq \wmin^{\mathrm{key}} \wedge \theta_j \in [\pi , 3\pi/2),\\
        3 & W_j \geq \wmin^{\mathrm{key}} \wedge \theta_j \in [3\pi/2 , 2\pi),
        \end{cases}
    \end{equation*}
    If $T_j = 1$, he applies the following discretisation map instead (see Fig.~\ref{fig:TestKeyRound}a):
    \begin{equation*}
        Z_j = \begin{cases}
        (\varnothing, 0) & W_j < \wmin \wedge \theta_j \in [0 , \pi/2), \\
        (\varnothing, 1) & W_j < \wmin \wedge \theta_j \in [\pi/2 , \pi), \\
        (\varnothing, 2) & W_j < \wmin \wedge \theta_j \in [\pi , 3\pi/2), \\
        (\varnothing, 3) & W_j < \wmin \wedge \theta_j \in [3\pi/2 , 2\pi), \\
        0 & \wmin \leq W_j \leq \wmax \wedge \theta_j \in [0 , \pi/2),\\
        1 & \wmin \leq W_j \leq \wmax \wedge \theta_j \in [\pi/2 , \pi),\\
        2 & \wmin \leq W_j \leq \wmax \wedge \theta_j \in [\pi , 3\pi/2),\\
        3 & \wmin \leq W_j \leq \wmax \wedge \theta_j \in [3\pi/2 , 2\pi), \\
        \top & W_j > \wmax
        \end{cases}
    \end{equation*}
    Repeat step 1 to 2 for $N$ times.
    
    \item \textit{Parameter estimation}: For each $j \in [N]$, Bob announces $T_j$. If $T_j = 0$, Bob announces the rounds $j$ in which $Z_j = \varnothing$. For each $T_j = 1$, Alice announces $X_j$. Bob computes $C_j = f_{\mathrm{PE}}(T_j, X_j, Z_j)$ for some deterministic function $f_{\mathrm{PE}}$ where $C_j = \perp$ if and only if $T_j = 0$ and $C_j = \top$ if and only if $T_j = 1$ and $Z_j = \top$. If $\freq_{\vb{C}} \in \cF_{\tol}$, they continue to the next step, else they abort.
    
    \item \textit{Classical post-processing}: Bob sends Alice a syndrome of his data. They apply error correction and error verification. If the error verification passes, then they apply privacy amplification on their data to obtain the final secret key.
\end{enumerate}
\section{Security analysis} \label{sec: security}
In this section, we shall present the security analysis of DM-CV-QKD protocols. As defined in Section~\ref{subsec: sec def}, to prove the security of a QKD protocol, we need to show that it satisfies the completeness, correctness and secrecy criteria.

\subsection{Correctness}
The proof of the correctness of the protocol is pretty straightforward. The correctness of the protocol is guaranteed by the error verification step, which is performed after the error correction step. In the error correction step, Bob sends the syndrome of his key so that Alice can produce a guess $\hat{\vb{Z}}$ of his key, $\vb{Z}$. The purpose of the error verification step is to verify that the error correction step has successfully enables Alice to correctly guess Bob's key. If the error correction step fails (i.e., when $\hat{\vb{Z}} \neq \vb{Z}$), we want the error verification step to detect this failure with high probability. This can be done using the property of two-universal hashing.

In the error verification step, Bob randomly picks a hash function $f_{\mathrm{EV}}$ from the two-universal family and calculate the hash of his key, $f_{\mathrm{EV}}(\vb{Z})$. He would then announce the hash function  that he chose and the corresponding hash value to Alice using the public but authenticated classical communication channel. Alice would then compute the hash value of her guess of Bob's key, $f_{\mathrm{EV}}(\hat{\vb{Z}})$. If $f_{\mathrm{EV}}(\hat{\vb{Z}}) \neq f_{\mathrm{EV}}(\vb{Z})$, then they abort the protocol.

The correctness of the protocol is due to the property of the two-universal hash function. Let $\ell_{\mathrm{EV}}$ be the length of the hash that Bob announced. Then, by the property of two-universal hash functions, we have
\begin{equation}
    \Pr[f_{\mathrm{EV}}(\hat{\vb{Z}}) \neq f_{\mathrm{EV}}(\vb{Z}) | \hat{\vb{Z}} \neq \vb{Z}] \leq \frac{1}{2^{\ell_{\mathrm{EV}}}}
\end{equation}
Thus, by taking $\ell_{\mathrm{EV}} \geq \log_2(1/\ecor)$, and noting that $\hat{\vb{Z}} \neq \vb{Z}$ implies $\vb{K}_A\neq\vb{K}_B$, then
\begin{multline}
    \Pr[\vb{K}_A\neq\vb{K}_B\land\Omega]\\
    \leq\Pr[\hat{\vb{Z}} \neq \vb{Z}]\Pr[\Omega | \hat{\vb{Z}} \neq \vb{Z}] \leq \ecor,
\end{multline}
which proves that the protocol is $\ecor$-correct.

\subsection{Completeness}
To prove completeness, we recall that the protocol can be aborted in either the parameter estimation step or in the error verification step. Therefore, we can write
\begin{equation}
    \Pr[\text{abort}] \leq \Pr[\text{abort in PE}] + \Pr[\text{abort in EV}],
\end{equation}
where $\Pr[\text{abort in PE}]$ is the probability that the protocol is aborted in the parameter estimation step and $\Pr[\text{abort in EV}]$ is the probability that the protocol is aborted in the error verification step. The latter can be upper bounded also
\begin{equation}
    \Pr[\text{abort in EV}] \leq \Pr[\text{EC fail}],
\end{equation}
where $\Pr[\text{EC fail}]$ is the probability that the error correction algorithm fails to make Alice's key identical to Bob's key.

Therefore, we can choose the completeness parameter $\ecom$ to be equal to
\begin{equation}
    \ecom = \ecom^{\mathrm{PE}} + \ecom^{\mathrm{EC}},
\end{equation}
where $\ecom^{\mathrm{PE}} \geq \Pr[\text{abort in PE}]$ and $\ecom^{\mathrm{EC}} \geq \Pr[\text{EC fail}]$.

While there is an information-theoretic bound on $\Pr[\text{EC fail}]$ based on the error rate in the honest implementation, the error correction algorithm that saturates this bound is difficult to implement in practice~\cite{renes2012oneshot}. On the other hand, it is typically difficult to prove an upper bound on the failure probability of practical error correction algorithms and one often relies on heuristics. Therefore, in this work, we shall focus on bounding the probability of aborting the protocol in the parameter estimation step.

To do that, we need to assume that the honest implementation produces a probability distribution $\vbf{p}$ on the score $C$. Then, for each $c \in \cC$, we fix $\zeta_c, \zeta'_c > 0$ such that the protocol is accepted if and only if
\begin{equation}
    p(c) - \zeta'_c \leq \freq(c) \leq p(c) + \zeta_c
\end{equation}
for all $c \in \cC$.

Now, we assume that for each $c \in \cC$, there exists $\ecom^{\mathrm{PE},c} > 0$ such that
\begin{equation}
    \Pr[ p(c) - \zeta'_c \leq \mathrm{freq}(c) \leq p(c) + \zeta_c] \geq 1 - \ecom^{\mathrm{PE},c}.
\end{equation}
We can construct such $\ecom^{\mathrm{PE},c}$ by taking the complement of the above event, and simplify using the union bound,
\begin{multline}
    \ecom^{\mathrm{PE}, c} \geq \Pr[\freq(c) \geq p(c) - \zeta'_c] \\
    + \Pr[\freq(c) \leq p(c) + \zeta_c].
\end{multline}
In other words, we want to find an upper bound on the sum of the probabilities of the two extreme events. To obtain such upper bound, we can consider the following bound on the binomial distribution, which was derived in Ref.~\cite{zubkov2013complete}
\begin{lemma}[Concentration inequality for binomial distribution~\cite{zubkov2013complete}] \label{lem: binomial bound}
Let $n \in \mathbb{N}$, $p \in (0,1)$ and let $X$ be a random variable distributed according to $X \sim \mathrm{Binomial}(n,p)$. Then, for any integer $k$ such that $0 \leq k < n$, we have
\begin{equation}
    F(n,p,k) \leq \Pr[X \leq k] \leq F(n,p,k+1),
\end{equation}
where
\begin{align*}
    D(q,p) &= q \ln \left( \frac{q}{p} \right) + (1-q) \ln \left( \frac{1-q}{1-p}\right),\\
    \Phi(a) &= \frac{1}{\sqrt{2\pi}} \int_{-\infty}^a \mathrm{d} x \, e^{-x^2/2},\\
    F(n,p,k) &= \Phi \left(\mathrm{sign}\left(\frac{k}{n} - p \right) \sqrt{2n D\left( \frac{k}{n}, p\right)} \right).
\end{align*}
\end{lemma}

To apply the above lemma, for each $c \in \cC$ and for each $i \in [N]$, we set the random variable $V_i^{(c)}$ as
\begin{equation}
    V_i^{(c)} =
    \begin{cases}
        1 & \text{if $C_i = c$}\\
        0 & \text{otherwise}.
    \end{cases}
\end{equation}
Then, $V_{\mathrm{tot}}^{(c)} = \sum_{i = 1}^N V_i^{(c)}$ is a binomial random variable which allows us to apply the lemma. For $\Pr[\freq(c) \geq p(c) - \zeta'_c]$, we can write
\begin{align}
    &\Pr[\freq(c) \geq p(c) - \zeta'_c] \nonumber\\
    &\leq \Pr[V_{\mathrm{tot}}^{(c)} \geq \floor{N (p(c) - \zeta'_c)}] \nonumber \\
    &= 1 - \Pr[V_{\mathrm{tot}}^{(c)} \leq \floor{N (p(c) - \zeta'_c)} - 1] \nonumber\\
    &\leq 1 - F(N, p(c), \floor{N (p(c) - \zeta'_c)} - 1).
\end{align}
Similarly, for $\Pr[\freq(c) \leq p(c) + \zeta_c]$ we write
\begin{align}
    &\Pr[\freq(c) \leq p(c) + \zeta_c] \nonumber\\
    &\leq \Pr[V_{\mathrm{tot}}^{(c)} \leq \ceil{N (p(c) + \zeta_c)}] \nonumber \\
    &\leq F(N, p(c), \ceil{N (p(c) + \zeta_c)} + 1).
\end{align}
Therefore, we can choose $\ecom^{\mathrm{PE},c}$ as
\begin{multline}
    \ecom^{\mathrm{PE},c} = 1 - F(N, p(c), \floor{N(p(c) - \zeta'_c)} - 1) \\ 
    + F(N, p(c), \ceil{N(p(c) + \zeta_c)} + 1).
\end{multline}
Finally, one can vary the choice of $\zeta'_c, \zeta_c > 0$ to maximise the secret key rate subject to
\begin{equation}
    \sum_{c \in \cC} \ecom^{\mathrm{PE}, c} \leq \ecom^{\mathrm{PE}}.
\end{equation}
Thus, any choice of $\zeta'_c$ and $\zeta_c$ that satisfies the above constraint will automatically satisfy the completeness criterion for the parameter estimation.

\subsection{Secrecy}
\subsubsection{Leftover hash lemma} \label{subsec: leftover hashing}
To prove the secrecy of the DM-CV-QKD protocol, we shall leverage on the leftover hash lemma. Before we do that, for a fixed $\eEAT \in (0,1)$, we consider the following two distinct cases
\begin{itemize}
    \item $\Pr[\Omega] \leq \eEAT$\\
    In this case, we can choose the secrecy parameter $\esec = \eEAT$ since the trace distance term $\Delta(\rho_{\vb{K}_B \vb{L} E | \Omega}, \tau_{\ell} \otimes \rho_{\vb{L} E|\Omega})$ in the definition of secrecy is upper bounded by one.
    \item $\Pr[\Omega] > \eEAT$\\
    In this case, we can use the leftover hash lemma to bound $\Delta(\rho_{\vb{K}_B \vb{L} E | \Omega}, \tau_{\ell} \otimes \rho_{\vb{L} E|\Omega})$ in terms of the conditional smooth min-entropy $H_{\min}^{\esmooth}(\vb{Z}|\vb{L}, E)_{\rho_{\vb{Z}\vb{L}E|\Omega}}$, which in turn can be lower bounded using GEAT while substituting $\Pr[\Omega]$ with $\eEAT$ in Theorem~\ref{thm: GEAT}.
\end{itemize}
\subsubsection{GEAT channels}
For our purpose, we consider the following GEAT channel $\cM_j: E_{j-1} L^{j-1} \rightarrow Z_j E_j L^j C_j$ for the $j$-th round, where the register $L_j$ stores the classical information accessible to Eve in the $j$-th round.
\begin{enumerate}
    \item Alice randomly generates $X_j$ and then prepares the state $\ket{\psi_{X_j}}_{Q_j}$.
    
    \item Eve applies her attack, given by the CPTP map $\cE_j$: $E_{j-1} Q_j \rightarrow E_j B_j$, where $B_j$ is the quantum signal received by Bob in the $j$-th round.
    
    \item Bob randomly generates $T_j \in \{0,1\}$ to decide whether a given round is a generation or test round. Bob measures the incoming signal using a heterodyne detection to obtain the outcome $Y_j$. He then performs the appropriate discretisation on $Y_j$ based on $T_j$ to obtain $Z_j$ (refer to Section~\ref{sec: protocol}).
    
    \item If $Z_j = \varnothing$ and $T_j = 0$, Bob sets $V_j = 1$, otherwise he sets $V_j = 0$. Next, Bob announces $T_j$ and $V_j$. If $T_j = 1$, Alice announces $X_j$. If $T_j = 1$, set $L^{\mathrm{test}}_j = X_j$. If $T_j = 0$, set $L^{\mathrm{test}}_j = \perp$. Finally, set $L_j = (V_j, T_j, L^{\mathrm{test}}_j)$. Note that $L_j$ is accessible to Eve. We can compute $C_j$ from $Z_j$ and $L_j$. Finally, we trace out $X_j$.
\end{enumerate}
The above channel satisfies the properties of GEAT channels in the sense that it satisfies the requirements that the register $C_j$ is a deterministic function of $Z_j$ and $L_j$ and also the no-signalling condition. To see that this is indeed the case, note that since Alice and Bob do not have an internal memory register (since we are working in device dependent framework), we can consider a channel $\cR_j$ where Eve can simply simulate their actions (i.e., Alice's state preparation and Bob's measurement) and trace out their private registers in the end. By construction, such channel will always satisfy $\cR_j = \Tr_{Z_j C_j} \circ \cM_j$. Thus, we can use these GEAT channels to apply the generalised EAT.

\subsubsection{Constructing min-tradeoff functions via dimension reduction}
Next, we construct the min-tradeoff function. As discussed in Section~\ref{subsec: GEAT}, to construct the min-tradeoff function, it is customary to find a function $g$
that satisfies the property stated in Eq.~\eqref{eq: min-tradeoff g}. We realise that function $g(q)$ is essentially a lower bound on the conditional entropy subjected to the statistics from the test rounds. 

It is not easy to construct such function directly as the dimension of Bob's Hilbert space is infinite. Instead, it is more convenient to consider a \textit{virtual scenario} where after Eve performs her attack (at the end of Step 2 of the GEAT channel), there is another completely positive, trace-non-increasing (CPTNI) map $\cT_j$ that truncates Bob's state into the space spanned by the Fock states $\{\ket{0}, ..., \ket{\nmax}\}$. More precisely, the map $\cT_j$ is given by
\begin{equation}
    \cT_j[\rho_{B_j E_j}] = (N_0 \otimes \1_{E_j}) \rho_{B_j E_j} (N_0 \otimes \1_{E_j}),
\end{equation}
where $N_0 = \sum_{n = 0}^{\nmax} \ketbra{n}{n}$ is the projector to the subspace spanned by $\{\ket{0}, ..., \ket{\nmax}\}$.

We want to find a function $g(q)$ that has the following form:
\begin{equation}
    g(q) = \tilde{g}(q) - g_{\corr}^{(\nu_c)}(q_{\top}),
\end{equation}
where the function $\tilde{g}(q)$ has similar properties as the function $g(q)$, except that it is evaluated on the output state of the truncated GEAT channel, where we applied the CPTNI map $\cT_j$ at the end of Step 2 (instead of the actual GEAT channel). In other words, the function $\tilde{g}: \cP_{\cC \setminus \{\perp\}} \rightarrow \mathbb{R}$ is an affine function that gives the lower bound on the conditional von Neumann entropy evaluated on the truncated state. The function $g_{\corr}^{(\nu_c)}(q_{\top})$ gives the appropriate correction term and it only depends on $q_{\top}$, which is the probability that $C_j = \top$. Importantly, unlike the function $g(q)$, the function $\tilde{g}(q)$ can be computed using a finite-dimensional optimisation.

More precisely, to construct the function $g(q)$, we can use the following dimension reduction theorem, which we prove in Appendix~\ref{annex: dimension reduction}. 
\begin{widetext}
    \begin{thm}[Dimension reduction] \label{thm: dimension reduction}
    For a fixed $\nmax \in \mathbb{N}$, we denote $\kappa = \Gamma(\nmax + 2, 0) / \Gamma(\nmax+2, \wmax)$, where $\Gamma$ denotes the upper incomplete gamma function, and we denote $N_0 = \sum_{n = 0}^{\nmax} \ketbra{n}{n}$. We let $\rho_{ABE} \in \mathsf{D}(\cH_A \otimes \cH_B \otimes \cH_E)$ be a normalised quantum state and $\tilde{\rho}_{ABE} = (\1_A \otimes N_0 \otimes \1_E) \rho_{ABE} (\1_A \otimes N_0 \otimes \1_E)$ be the truncated version of the state. For each $c \in \cC \setminus \{\perp\}$, we let $\Pi_c \in \mathsf{L}(\cH_A \otimes \cH_B)$ be the projector that corresponds to the score $c$ and $\tilde{\Pi}_c = (\1_A \otimes N_0) \Pi_c (\1_A \otimes N_0)$ be its truncated version. We let $\nu_c \in (0, 2/3\kappa), \nu_L \in (0, (5+\sqrt{5})/10\kappa]$, and $\nu_U \in (0, (5-\sqrt{5})/10\kappa]$ be given. Finally, we let $\tilde{\rho}_{ZLE}$ be the state resulting from applying Bob's measurement on the truncated state $\tilde{\rho}_{ABE}$, storing the output on the classical register $Z$ (with dimension $d_Z$), making the appropriate announcement $L$, and then tracing out Alice's subsystem. Then, we have
    \begin{equation*}
        g(q) = \tilde{g}(q) - g_{\corr}^{(\nu_c)}(q_{\top}),
    \end{equation*}
    where $\tilde{g}(q)$ is the linearisation of the dual $\tilde{g}'(q)$
    \begin{equation} \label{eq: finite dimensional optimisation}
    \begin{split}
        \tilde{g}'(q) \leq \inf_{\tilde{\rho}} \, &H(Z|L,E)_{\tilde{\rho}_{ZLE}}\\
        \text{s.t.} \, &\tilde{\rho}_{AB} \succeq 0\\
        & \Tr[\tilde{\rho}_{AB}] \leq 1,\\
        & \Tr[\tilde{\rho}_{AB}] \geq 1 - \kappa q_{\top},\\
        & \Delta(\tilde{\rho}_{A}, \Tr_Q \ketbra{\Psi}{\Psi}_{AQ}) \leq \frac{1}{2}\kappa q_{\top}\\
        & \Tr[\tilde{\rho}_{AB} \tilde{\Pi}_c] \leq q_c - \xi_U(q_{\top}) \quad \forall c \in \cC \setminus \{\perp\}, \\
        & \Tr[\tilde{\rho}_{AB} \tilde{\Pi}_c] \geq q_c - \xi_L(q_{\top}) \quad \forall c \in \cC \setminus \{\perp\},
    \end{split}
\end{equation}
and where we define the following correction terms
\begin{align*}
    g^{(\nu_c)}_{\corr}(q_{\top}) &= m_{\corr} (q_{\top} - \nu_c) + c_{\corr},\\
    \xi_L(\nu) &= \begin{cases}
        \kappa\nu + 2\sqrt{\kappa\nu(1-\kappa\nu)} & \mathrm{if} \quad \nu< \frac{5+ \sqrt{5}}{10\kappa}\\
        \frac{1+\sqrt{5}}{2} & \mathrm{if} \quad \nu \geq \frac{5+ \sqrt{5}}{10\kappa}     
    \end{cases}\\
    \xi_U(\nu) &= \begin{cases}
        \kappa\nu - 2 \sqrt{\kappa\nu(1-\kappa\nu)}  &\mathrm{if} \quad \nu < \frac{5 - \sqrt{5}}{10\kappa}\\
        \frac{1 - \sqrt{5}}{2} \qquad &\mathrm{if} \quad \nu \geq \frac{5 - \sqrt{5}}{10\kappa}
    \end{cases}
\end{align*}
with $w_c = \kappa \nu_c$
\begin{align*}
    \delta_c &= \frac{1}{2\sqrt{2}}\Bigg[\sqrt{w_c (2 - w_c) + w_c\sqrt{w_c (4 - 3 w_c)}} + \sqrt{w_c (2 - w_c) - w_c \sqrt{w_c (4 - 3 w_c)}} \Bigg]\\
    m_{\corr} &= \Bigg[\frac{3 w_c + \sqrt{w_c (4 - 3 w_c)}}{\sqrt{(2 - w_c) + \sqrt{w_c (4 - 3 w_c)}}} - \frac{3w_c - \sqrt{w_c (4 - 3w_c)}}{\sqrt{(2 - w_c) - \sqrt{w_c (4 - 3w_c)}}} \Bigg] \\
    c_{\corr} &= \delta_c \log_2 d_Z + (1 + \delta_c) h_2 \left(\frac{\delta_c}{1 + \delta_c}\right).
\end{align*}
and linearised terms
\begin{align*}
    \xi_U(q_{\top}) &\geq \hat{\xi}_U^{(\nu_U)}(q_{\top}) = \left(1 - \frac{(1- 2 \kappa \nu_U)}{\sqrt{\kappa \nu_U (1 - \kappa \nu_U)}} \right) \kappa q_{\top} - \sqrt{\frac{\kappa \nu_U}{1-\kappa \nu_U}},\\
    \xi_L(q_{\top}) &\leq \hat{\xi}_L^{(\nu_L)}(q_{\top}) = \left(1 + \frac{(1- 2 \kappa \nu_L)}{\sqrt{\kappa \nu_L (1 - \kappa \nu_L)}} \right) \kappa q_{\top} + \sqrt{\frac{\kappa \nu_L}{1-\kappa \nu_L}},   
\end{align*}
replacing their respective terms in the dual.
\end{thm}
\begin{proof}
    Refer to Appendix \ref{annex: dimension reduction}.
\end{proof}

\end{widetext}

Theorem~\ref{thm: dimension reduction} allows us to efficiently construct the function $g(q)$ since constructing $\tilde{g}(q)$ involves an optimisation over finite dimensional quantum states. This can be done using numerical techniques~\cite{araujo2023quantum, winick2018reliable}. In particular, the above optimisation can be solved using SDP and since the constraints are linear in $q$, the dual solution of the SDP will readily provide us with an affine function. We shall show the explicit SDP and its corresponding dual solution in the next two subsections. This completes the min-tradeoff construction.

\subsubsection{SDP method to bound the conditional entropy}\label{sec: SDP to bound entropy}
The minimisation of the conditional von Neumann entropy in Eq.~\eqref{eq: finite dimensional optimisation} involves non-linear objective function which can be difficult to solve using readily available numerical solvers. Thankfully, one can lower bound the conditional von Neumann entropy as SDP using the method proposed in Ref.~\cite{araujo2023quantum} which is based on the bound in Ref.~\cite{brown2021device}. Before we use the bound, we first simplify the minimisation by only keeping the term corresponding to the generation rounds and the case in which the outcome is not discarded
\begin{align}
    H(Z|L,E)_{\tilde{\rho}} &\geq (1-\gamma) \Pr[Z \neq \varnothing| T = 0] \nonumber\\
    &\hspace{1cm} H(Z|T = 0, Z \neq \varnothing, E)_{\tilde{\rho}} \\
    &=: H(Z|E)_{\rho^*},
\end{align}
where $\rho^*_{ZE} = \Tr_T[\Pi_{\pass} \tilde{\rho}_{ZTE} \Pi_{\pass}]$, with $\Pi_{\pass} = \sum_{z \neq \varnothing} \ketbra{z}{z}_Z \otimes\ketbra{0}{0}_T \otimes  \1_E$, is the sub-normalised state corresponding to the case $T = 0$ and $Z \neq \varnothing$.

We then convert the conditional von Neumann entropy into quantum relative entropy. Let $\rho^*_{ZE}$ be a finite dimensional (sub-normalised) cq-state and $\rho^*_E = \Tr_Z[\tilde{\rho}^*_{ZE}]$. Then,
\begin{equation}
    H(Z|E)_{\rho^*} = -D(\rho^*_{ZE}|| \1_Z \otimes \rho^*_{E})
\end{equation}
Then, Ref.~\cite{brown2021device} formulated an upper bound on the quantum relative entropy (which gives a lower bound on the conditional von Neumann entropy) based on Gauss-Radau quadrature. Let $m \geq 2$ be an integer and let $(w_i, t_i)$ be weights and nodes of the Gauss-Radau quadrature with fixed node at $t = 1$. Let $\Lambda_i \in \mathsf{B}(\cH_{ZE})$ be an arbitrary bounded operator in the Hilbert space $\cH_{ZE}$ for all $i \in [m]$. The conditional von Neumann entropy is lower bounded by
\begin{multline}
    \inf_{\{\Lambda_i\}_i, \tilde{\rho}} \sum_{i = 1}^m \frac{w_i}{t_i \ln 2} \\
    \Tr\Bigg[\rho^*_{Z E} \left(\1_{ZE} + \Lambda_i + \Lambda_i^\dagger + (1-t_i) \Lambda_i^\dagger \Lambda_i\right)\\ + t_i
    (\1_Z \otimes \rho^*_E) \Lambda_i \Lambda_i^\dagger\Bigg]
\end{multline}
Since $\cH_Z$ is finite dimensional, the above bound can be simplified by writing $\Lambda_i = \sum_{z,z'} \ketbra{z}{z'} \otimes \Lambda_i^{(z,z')}$. Furthermore, we write
\begin{multline*}
    \tilde{\rho}_{ZTE} = \sum_{z, t} \Pr[T = t] \,\ketbra{z}{z}_Z \otimes \ketbra{t}{t}_T \\ \otimes \Tr_{AB}[\tilde{\rho}_{ABE} (\1_A \otimes P_{z|t} \otimes \1_{E})]
\end{multline*}
where $P_{z|t}$ is Bob's POVM element. From our definition of the state $\rho^*_{ZE}$, this results in 
\begin{multline}
    H(Z|E)_{\rho^*} \geq (1-\gamma) \inf_{\tilde{\rho}, \{\Lambda_{i,z}\}_{i,z}} \sum_{i = 1}^m \frac{w_i}{t_i \ln 2}
   \Tr\Bigg[\tilde{\rho}_{ABE} \\ \Bigg(P_{PS} \otimes \1_E + \sum_{z \neq \varnothing} (\1_A \otimes P_{z|0}) \otimes (\Lambda_{i,z} + \Lambda_{i,z}^\dagger + (1-t_i) \Lambda_{i,z}^\dagger \Lambda_{i,z}) \\+ t_i  
    P_{PS} \otimes \Lambda_{i,z} \Lambda_{i,z}^\dagger \Bigg)\Bigg],
\end{multline}
where $P_{PS} = \sum_{z \neq \varnothing} \1_A \otimes P_{z|0}$

Since the optimisation is over both $\tilde{\rho}$ and $\Lambda$, the problem is still non-linear. However, as shown in Ref.~\cite{araujo2023quantum}, the problem can be linearised by defining the map $\Xi[M] = \Tr_E[\tilde{\rho}_{ABE} \cdot (\1_{AB} \otimes M^T)]$ for any operator $M$, which has the following properties
\begin{equation}
\begin{split}
    \Xi[M^\dagger] &= \Xi[M]^\dagger\\
    \Tr[\tilde{\rho}_{ABE} \cdot (K_{AB} \otimes M^T)] &= \Tr[K_{AB} \cdot \Xi(M)]     
\end{split}
\end{equation}
for any operator $K_{AB}$. We then define the matrices
\begin{equation}
    \begin{split}
        \sigma &= \Xi[\1]\\
        \omega_{i, z} &= \Xi[\Lambda_{i,z}]\\
        \eta_{i, z} &= \Xi[\Lambda_{i,z}^
        \dagger \Lambda_{i,z}]\\
        \theta_{i,z} &= \Xi [\Lambda_{i,z} \Lambda_{i,z}^\dagger]
    \end{split}
\end{equation}
and consider block moment matrices. The resulting lower bound of the conditional entropy is given by
\begin{widetext}
\begin{equation} \label{eq: final SDP}
\begin{split}
    \inf_{\sigma, \{\omega_{i,z}, \eta_{i,z}, \theta_{i,z}\}_{i,z}, \zeta_1, \zeta_2}  & (1-\gamma) \sum_{i = 1}^m \frac{w_i}{t_i \ln 2} \Tr\left[\sigma  P_{PS} + \sum_{z \neq \varnothing}  (\1_A \otimes P_{z|0}) (\omega_{i,z} + \omega_{i,z}^\dagger + (1-t_i) \eta_{i,z}) + t_i \theta_{i,z}  P_{PS} \right] \\
    \text{s.t.} \qquad & \Tr[\sigma] \leq 1\\
    & \Tr[\sigma] \geq 1 - \kappa q_{\top}\\
    & \Tr[\zeta_1 + \zeta_2] \leq \kappa q_{\top}\\ 
    & \Tr[\sigma \tilde{\Pi}_c] \leq q_c - \xi_U(q_{\top}) \quad \forall c \in \cC \setminus \{\perp\}, \\
    & \Tr[\sigma \tilde{\Pi}_c] \geq q_c - \xi_L(q_{\top}) \quad \forall c \in \cC \setminus \{\perp\},\\
    & \begin{pmatrix}
        \sigma & \omega_{i, z}\\
        \omega_{i, z}^\dagger & \eta_{i, z}
    \end{pmatrix} \succeq 0 \quad \forall i, z\\ 
    &\begin{pmatrix}
        \sigma & \omega_{i, z}^\dagger\\
        \omega_{i, z} & \theta_{i, z}
    \end{pmatrix} \succeq 0 \quad \forall i, z\\
    &\begin{pmatrix}
        \zeta_1 & \Tr_B[\sigma] - \Tr_Q [\ketbra{\Psi}{\Psi}]\\
        \Tr_B[\sigma] - \Tr_Q[\ketbra{\Psi}{\Psi}] & \zeta_2
    \end{pmatrix} \succeq 0
\end{split}
\end{equation}
\end{widetext}

\subsubsection{Dual solution}
We take the optimisation~\eqref{eq: final SDP} and associate to each constraint that depends on $q$ with a dual variable $\lambda$
\begin{align*}
    \lambda_{\text{norm}}: \quad& 1 - \Tr[\sigma] \leq \kappa q_{\top}\\
     \lambda_{\text{dist}} : \quad&\Tr[\zeta_1 + \zeta_2] \leq \kappa q_{\top}\\ 
    \lambda_c^U: \quad &\Tr[\sigma \tilde{\Pi}_c] \leq q_c - \xi_U(q_{\top})\\ 
    \lambda^L_c: \quad & \Tr[\sigma \tilde{\Pi}_c] \geq q_c - \xi_L(q_{\top})  
\end{align*}

Then, the dual solution of the SDP~\eqref{eq: final SDP} will have the form
\begin{multline}
    \tilde{g}'(q) \geq \varphi - \lambda_{\text{norm}} \kappa q_{\top} - \lambda_{\text{dist}} \kappa q_{\top}\\
    + \sum_{c \neq \perp, \top} \left(- \lambda_c^U(q_c - \xi_U(q_{\top})) + \lambda_c^L(q_c - \xi_L(q_{\top})) \right) \\- \lambda_{\top}^U (q_{\top} - \xi_U(q_{\top})) + \lambda_{\top}^L ( q_{\top} - \xi_L(q_{\top})),
\end{multline}
where $\varphi$ contains all the dual variables that are associated to the non-statistical constraints and the forms are chosen such that all dual $\vbf{\lambda} \geq0$.
The dual solution is linearised by replacing the terms $\delta$, $\xi_L$ and $\xi_U$ with linear terms
\begin{equation}
\label{eq: linearisation method}
\begin{split}
    \xi_U(q_{\top}) &\geq m_U q_{\top} + c_U\\
    \xi_L(q_{\top}) &\leq m_L q_{\top} + c_L,
\end{split}
\end{equation}
as presented in Eq.~\eqref{eqn: linear xi_U} and Eq.~\eqref{eqn: linear xi_L} respectively.
Then, we have
\begin{multline}
\label{eq: dual function main}
    \tilde{g}(q) \geq \varphi' + \sum_{c \neq \perp, \top} (-\lambda^U_c + \lambda_c^L) q_c \\
    + q_{\top} \Bigg( -\kappa \lambda_{\text{norm}} - \kappa \lambda_{\text{dist}} - \lambda^U_{\top} + \lambda^L_{\top} \\+ \sum_{c \neq \perp} (\lambda_c^U m_U - \lambda_c^L m_L)\Bigg),
\end{multline}
with
\begin{equation}
    \varphi' = \varphi + \sum_{c \neq \perp} (\lambda_c^U c_U - \lambda_c^L c_L).
\end{equation}
Finally, considering the $g_{\corr}$ term, we have
\begin{equation}
    g(q) \geq \Phi + \sum_{c \neq \perp} \lambda'_c \, q_c
\end{equation}
where
\begin{align}
    \Phi &= \varphi' - c_{\corr} + m_{\corr} \nu_c \nonumber\\
    &= \varphi + \sum_{c \neq \perp} (\lambda_c^U c_U + \lambda_c^L c_L) \nonumber - c_{\corr} + m_{\corr} \nu_c ,
\end{align}
and
\begin{widetext}
    \begin{equation}
    \lambda'_c =
    \begin{cases}
        \lambda_c^L - \lambda_c^U & \text{if $c \neq \perp, \top$},\\
        -\lambda_{\top}^U + \lambda_{\top}^L - \kappa \lambda_{\text{norm}} - \kappa \lambda_{\text{dist}} + \sum_{c \neq \perp} (-\lambda_c^L m_L + \lambda_c^U m_U) - m_{\corr} & \text{if $c = \top$}.
    \end{cases}
\end{equation}
\end{widetext}
Denote $\lambda'_{\max} := \max_{c} \lambda'_c$ and $\lambda'_{\min} := \min_c \lambda'_c$. We have the following min-tradeoff function
\begin{equation} \label{eq: min-tradoff conversion}
\begin{split}
    f(\vbf{\hat{e}}_c) &= \Phi -  \frac{(1-\gamma)}{\gamma} \lambda'_{\max} + \frac{\lambda'_c}{\gamma}, \qquad  c \in \cC \setminus \{\perp\},\\
    f(\vbf{\hat{e}}_{\perp}) &= \Phi + \lambda'_{\max}.
\end{split}
\end{equation}
The min-tradeoff function $f$ has the following properties
\begin{align}
    \Max[f] &= \Phi + \lambda'_{\max}\\
    \Min_{\Sigma}[f] &\geq \Phi + \lambda'_{\min}\\
    \Var[f] &\leq \frac{(\lambda'_{\max} - \lambda'_{\min})^2}{\gamma}
\end{align}

Having constructed the min-tradeoff function $f$ and calculate the its corresponding $\Max, \Min_{\Sigma}$, and $\Var$, we can then apply the GEAT to bound the conditional smooth min-entropy $H_{\min}^{\esmooth}(\vb{Z}|\vb{L}, E)_{\rho_{\vb{Z}\vb{L}E|\Omega}}$. Combined with the argument in Sec~\ref{subsec: leftover hashing}, this completes the proof for the secrecy of the protocol.
\section{Simulation} \label{sec: simulation}
To study the performance of the protocol which is secure against coherent attacks, we simulated the key generation rate of the protocol with different number of signals, $N$, sent.
For simplicity, an ideal error correcting code is utilised, which is assumed to correct errors by communicating information at the Shannon limit, with information loss $NH(X_j|Z_j)_{\rho^*}$ and without failure, $\ecom^{\mathrm{EC}}=0$.
We consider a protocol run with Bob performing ideal heterodyne detection with zero excess noise, $\chi=0$, and unit detection efficiency.
More details of the model can be found in Appendix~\ref{annex: model}.
Based on preliminary optimisation, we choose the discretisation $\tau_{\min}^{\mathrm{key}}=0.6$, $\tau_{\min}=1.5$ and $\tau_{\max}=\sqrt{20}$.
Furthermore, we fix the completeness parameter for parameter estimation as $\ecom^{\mathrm{PE}}=$ \SI{1e-10}{}, correctness parameter as $\ecor=$ \SI{1e-15}{}, secrecy parameter as $\esec=$ \SI{1e-6}{}, and $n_{\max}=$ 12.
The correctness parameter results in $l_{EV}=$ \SI{50}{} bits error verification string. 

To obtain the dual $\tilde{g}(q)$, i.e. values for $\varphi$, and dual parameters $\vbf{\lambda}$, one can solve the SDP in Eqn.~\eqref{eq: final SDP} with any set of trial parameters $q_{\mathrm{dual}}$.
The dual parameters and $\varphi$ obtained can compute a valid dual function $\tilde{g}^{(q_{\mathrm{dual}})}(q)$ of the form in Eqn.~\eqref{eq: dual function main}.
Importantly, this dual function is a lower bound of the dual function with trial parameters matching the actual parameters, $\tilde{g}^{(q_{\mathrm{dual}})}(q)\leq\tilde{g}^{(q)}(q)$, by definition of the dual.
In general, we are free to choose the trial parameters $q_{\mathrm{dual}}$ to optimise the key rate (note that $q_{\mathrm{dual}}=q$ may not be the optimal choice for EAT).
However, to simplify, we would consider only the set of parameters $q_{\mathrm{dual}}$ that can be generated from a run of the protocol with the same model as described in the previous paragraph, except that the excess noise can be $\chi_{\mathrm{dual}}\neq\chi$.

\begin{figure}
    \centering
    \includegraphics[width=0.47\textwidth]{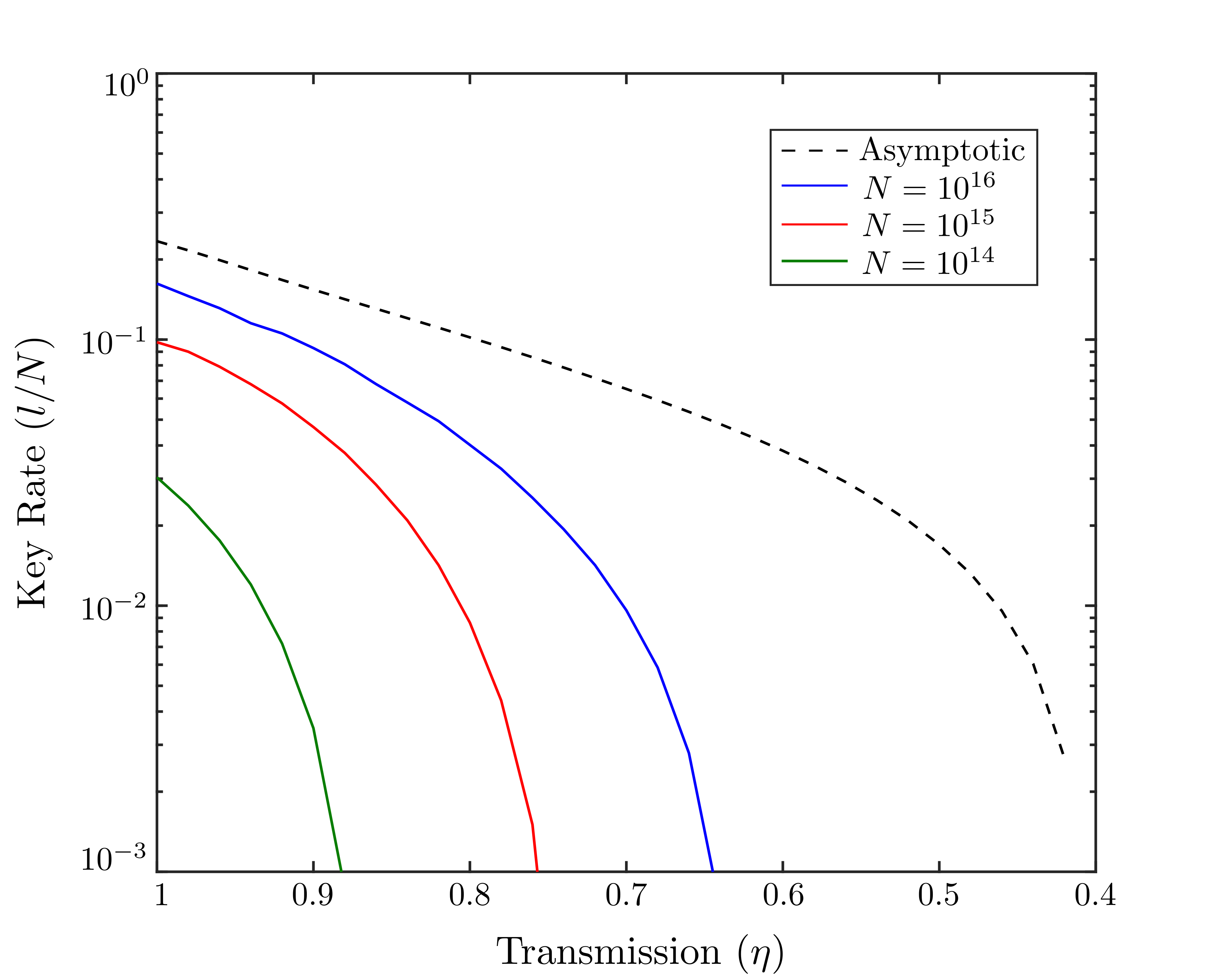}
    \caption{Plot of key rate against transmission, for $N=$ \SI{1e14}{}, \SI{1e15}{}, \SI{1e16}{}, with the asymptotic key rate plot in dotted line.}
    \label{fig:Simulation}
\end{figure}

The key rate, $\ell/N$, where $\ell$ is the length of the key, is then optimised over the amplitude $\alpha$, the EAT parameter $\beta$, the linearisation parameters $\nu_c$, $\nu_L$, and $\nu_U$, the trial parameter $\chi_{\mathrm{dual}}$, and the various components of the completeness parameter, $\ecom^{\mathrm{PE},c}$.
Fig.~\ref{fig:Simulation} shows the optimised key rate for $N=$ \SI{1e14}{} to \SI{1e16}{}, and the asymptotic key rate.
There is a significant penalty from finite size effect, with the key rate for finite block size being much lower than the asymptotic key rate.

\section{Discussion and conclusion} \label{sec: conclusion}
The results show a significant gap between the finite size key rate and asymptotic key rates, even for large $N=$ \SI{1e16}{}. There may be a couple of explanations for this. Firstly, it is possible that our optimisation heuristics for the choice of $q_{\mathrm{dual}}$ might be far from optimal since we only consider min-tradeoff functions that are parameterised by $\chi_{\mathrm{dual}}$ (see Appendix~\ref{annex: min-tradeoff} for more details). While it is possible to consider a more exhaustive optimisation for $q_{\mathrm{dual}}$ in principle, this is challenging to implement in practice since it is computationally expensive to solve the SDP in Eqn.~\eqref{eq: final SDP}.
This might potentially be improved by adopting the non-symmetric cone approach to bounding conditional von Neumann entropy~\cite{Lorente2024quantum}.

It is also possible that this variant of DM-CV-QKD protocol has significant finite-size penalties when analysed using GEAT. This may be the consequence of the linearisation terms. In general, the probability of exceeding $\tau_{\max}$ is small, indicating a low probability of having large energy. As such, the choice of linearisation parameters $\nu_U$ and $\nu_L$ would be correspondingly small as well to have a small gap between the actual value (e.g. $\xi_L$) and linearised value (e.g. $\hat{\xi}_L^{(\nu_L)}$). Since this term is present in the denominator of the gradient, it leads to a large gradient and causes the penalty from GEAT to be significant. As the finite-size penalty in both EAT and GEAT are very sensitive to the gradient of the min-tradeoff function, we expect that the huge finite-size penalty is a feature of the DM-CV-QKD protocol which estimates the dimension reduction correction using the probability of exceeding $\tau_{\max}$. Therefore, to improve the finite-size penalty, we may require a new technique that is less sensitive to the gradient of the min-tradeoff function. Alternatively, we may require another method to estimate the dimension reduction corrections. A possible approach is to investigate continuity bounds that are based on fidelities instead of trace distances. For example, in the case for quantum-classical states (instead of the case for classical-quantum states that we need here), it is known that the continuity bound based on fidelity~\cite{liu2023lipschitz} has a better correction term. Another potential room for improvement is to find corrections to the statistical constraints that have gentler gradients at $q_{\top} \approx 0.$ Recently, Ref.~\cite{arqand2024generalized} proposed the R\'enyi entropy version of the GEAT. Interestingly, this version of the entropy accumulation theorem does not require the usual analysis via affine min-tradeoff functions, which removes the necessity of linearising the correction terms. This approach also removes the need to find the best $q_{\mathrm{dual}}$, which might potentially reduce the minimum block size $N$.

\begin{figure}
    \centering
    \includegraphics[width=0.47\textwidth]{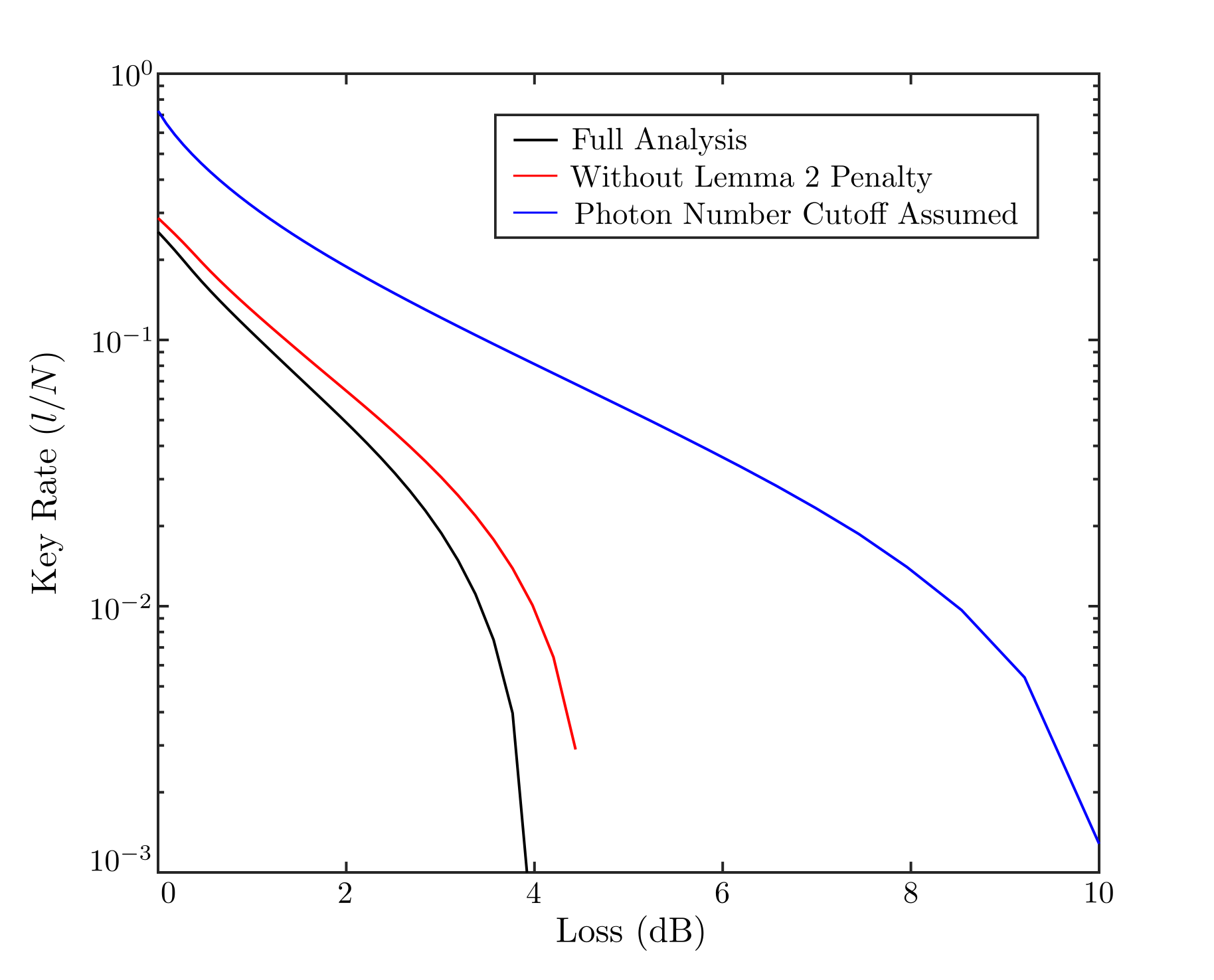}
    \caption{Plot of asymptotic key rate against loss with (1) full analysis (shown by the black curve), (2) removal of the penalty term from the continuity of conditional entropy in Lemma~\ref{lemma: continuity entropy} in Appendix~\ref{annex: dimension reduction} (shown by the red curve), and (3) additional removal dimension reduction corrections to the SDP constraints (shown by the blue curve).}
    \label{fig:Penalty_Compare}
\end{figure}

The asymptotic key rate also appears to be significantly worse than in Refs.~\cite{upadhyaya2021dimension, bauml2023security}, where the loss tolerance is over \SI{40}{\decibel} ($>200$~km with perfect detector) compared to around \SI{4}{\decibel} here. This is likely attributable to both the number of nodes  $m$ (for the Gauss-Radau quadrature approximation) selected in SDP optimisation, and the dimension reduction penalty (we note that the asymptotic key rate evaluated in Ref.~\cite{bauml2023security} does not account for the dimension reduction penalty).

In our simulation, we chose $m=4$ for the Gauss-Radau quadrature approximation to reduce the computation time. This choice may not be tight and may lead to a loss of key rate. This can be improved by adopting the non-symmetric cone approach for conditional von Neumann entropy~\cite{Lorente2024quantum} which avoids the Gauss-Radau approximation. Besides the penalty term $g_{\mathrm{corr}}^{(\nu_c)}(q_{\top})$, the corrections in the SDP constraints may not be tight either. In particular, the corrections to the statistical constraints do not take the structure of the measurement operators into account, and hence can be overly pessimistic. Fig.~\ref{fig:Penalty_Compare} shows a plot of the asymptotic key rate with the removal of various penalties. The looseness of the SDP bounds from the dimension reduction penalty appears to reduce the loss tolerance of the protocol significantly (from more than \SI{10}{\decibel} to \SI{4}{\decibel}). Indeed, based on Fig.~\ref{fig:Penalty_Compare}, one can infer that most of the penalties can be attributed to the corrections to the SDP constraints since most of the reduction in the key rate is observed once these corrections are properly accounted for. As such, both would result in a lower key rate by directly reducing the conditional von Neumann entropy and increasing the set of feasible states in the SDP minimisation.

It is also remarkable that in Refs.~\cite{upadhyaya2021dimension, kanitschar2023finite, lin2020trusted}, the dimension reduction theorem does not introduce such significant penalties in the key rate. In these works, the parameter estimation step of the protocol is based on estimation of moments (i.e., the expectation values of the physical observables, such as position or momentum) rather than probability distribution. This seems to suggest that DM-CV-QKD protocols that are based on estimation of moments are more robust against the dimension reduction penalties. However, proof techniques that are based on entropy accumulation typically require us to work with probability distributions instead. Therefore, it may be interesting to explore new approaches that are not based on the currently known entropy accumulation theorems~\cite{dupuis2019entropy, dupuis2020entropy, metger2022generalised, metger2023security, arqand2024generalized}.

Comparing our work with that of Ref.~\cite{bauml2023security}, the protocol that is being considered in both Ref.~\cite{bauml2023security} and ours is very similar, with the exception that we allow post-selection in the key generation rounds. Indeed, the discretisation that we use in the test rounds is almost identical to the one considered in Ref.~\cite{bauml2023security}. Since the asymptotic key rates obtained in Ref.~\cite{bauml2023security} are close to the one obtained when considering the dimension reduction penalty~\cite{upadhyaya2021dimension}, we suspect that either a very large number of nodes $m$ is required for a tight bound on the conditional von Neumann entropy, or the protocol variant that monitors probability distribution instead of moments will incur higher dimension reduction penalty. This can be investigated by computing the conditional von Neumann entropy using the methods in Refs.~\cite{winick2018reliable, hu2022robust}.
Alternative methods of bounding the relative entropy can also be explored, such as that in Ref.~\cite{Lorente2024quantum}.

To summarise, we presented a complete security analysis of a variant of the DM-CV-QKD protocol using four coherent-states and heterodyne detection. Unlike previous works, our security analysis accounts for both finite-size effects as well as dimension reduction. To achieve this, we apply the newly proposed GEAT and we modify the dimension reduction theorem presented in Ref.~\cite{upadhyaya2021dimension} to suit the technical requirements of GEAT. Additionally, we also applied the Gauss-Radau quadrature approximation to the conditional von Neumann entropy to construct the min-tradeoff function, which we need to apply GEAT. Our method is versatile as it can be easily extended to other variants of DM-CV-QKD protocols, and it is also applicable to the scenario where the detector imperfection is characterised and trusted (similar to the one considered in Ref.~\cite{lin2020trusted}). Unfortunately, our result suggests that DM-CV-QKD protocols may suffer from severe finite-size effects when analysed using EAT or GEAT due to the large gradient of the min-tradeoff function as well as increased sensitivity to the dimension reduction penalties. While our bounds can potentially be improved, we leave this investigation for future works.

\section*{Note added}
While preparing this manuscript, Ref.~\cite{pascual2024improved} independently published an improved finite-size analysis of DM-CV-QKD based on the GEAT. In that work, the authors claimed that the four-state DM-CV-QKD protocol (similar to the one analysed in this work) only requires the minimum block size of the order $N \sim 10^8$, which is much smaller than what we have shown in this work. However, the result of Ref.~\cite{pascual2024improved} also assumes the photon number cutoff without properly accounting for the dimension reduction penalties -- which we have shown to be non-negligible. Indeed, the authors remarked that they have not successfully obtained a positive key rate while accounting for the dimension reduction penalties.

Contrary to Ref.~\cite{pascual2024improved}, this work presented a full security proof, which includes the analysis on the dimension reduction penalty. Our analysis is also applicable to the DM-CV-QKD protocol that employs post-selection, while Ref.~\cite{pascual2024improved} focused on the class of protocols without post-selection.

\begin{acknowledgments}
The authors acknowledge funding support from the National Research Foundation of Singapore (NRF) Fellowship grant (NRFF11-2019-0001) and NRF Quantum Engineering Programme 1.0 grant (QEP-P2 and QEP-P3). 
We also thank Ernest Y.-Z. Tan for discussion and his valuable feedback on an earlier version of this manuscript.

Wen Yu Kon and Charles Lim, in their present role at JPMorgan Chase \& Co contributed to this work for information purposes. This paper is not a product of the Research Department of J.P. Morgan, Singapore, or its affiliates. Neither J.P. Morgan, Singapore nor any of its affiliates make any explicit or implied representation or warranty and none of them accept any liability in connection with this paper, including, but limited to, the completeness, accuracy, reliability of information contained herein and the potential legal, compliance, tax or accounting effects thereof. This document is not intended as investment research or investment advice, or a recommendation, offer or solicitation for the purchase or sale of any security, financial instrument, financial product or service, or to be used in any way for evaluating the merits of participating in any transaction.

\end{acknowledgments}

\bibliography{references}

\onecolumngrid
\appendix
\newpage

\section{Proof of dimension reduction theorem} \label{annex: dimension reduction}
The goal of dimension reduction theorem is to reduce the original problem of constructing a min-tradeoff function which involves optimisation over infinite-dimensional quantum states into a simpler problem which only involves optimisation over finite-dimensional quantum states. The idea is similar to a similar dimension reduction theorem presented in Ref.~\cite{upadhyaya2021dimension}: we perform a truncation on Bob's quantum state and calculate the correction terms attributed to this truncation. Before calculating these correction terms, we first present the original, untruncated optimisation problem.

\subsection{The original optimisation problem}
To derive a lower bound on the conditional von Neumann entropy, it is easier to consider a virtual entanglement-based scenario that is equivalent to the actual protocol from the adversary's point-of-view. More precisely, in the virtual scenario, the adversary will have the same classical and quantum-side information and the honest parties will obtain the same classical registers at the end of the procedure.

The virtual entanglement-based protocol that we consider is as follows:
\begin{enumerate}
    \item Alice prepares the state
    \begin{equation}
        \ket{\Psi}_{AQ} = \frac{1}{2} \sum_{x = 0}^{3} \ket{x}_A \otimes \ket{\psi_x}_{Q}
    \end{equation}
    inside her secure lab. 
    \item Alice measures the register $A$ in the computational basis and stores the outcome in the classical register $X$.
    \item Alice sends the quantum register $Q$ to Bob via the untrusted quantum channel. We assigns the output of the quantum channel that Bob receives to the quantum register $B$.
    \item Upon receiving the quantum register $B$ from the untrusted quantum channel, Bob randomly assigns $T \in \{0,1\}$ with probability $\{1-\gamma, \gamma\}$. He then applies heterodyne measurement on the quantum state and, depending on $T$, he applies the appropriate discretisation to obtain the discretised outcome $Z$.
\end{enumerate}

It is easy to see that the above procedure is equivalent to the quantum communication phase of the DM-CV-QKD protocol that we consider in this work. Step 1--2 is equivalent to uniformly choosing a classical information $x$ and encoding it to the quantum state $\ket{\psi_x}$. On the other hand, step 3--4 is identical to the one performed in the actual protocol. Therefore, the above virtual scenario is indeed equivalent to the actual protocol that we are analysing.

The advantage of considering the entanglement-based scenario is that it allows us to delay Step 2 (Alice's measurement) since local measurements on different subsystems commute. With this delayed measurement, we can consider the quantum state $\rho_{AB}$, shared between Alice and Bob, just before their respective measurement. In general, the Hilbert space $\cH_B$, associated with the register $B$ is infinite-dimensional as we do not restrict what Eve can do in the quantum channel. On the other hand, as the quantum register $A$ is generated and stored in Alice's secure lab, we have the constraint on Alice's marginal state:
\begin{equation}
    \rho_{A} = \Tr_{B}[\rho_{AB}] = \Tr_{Q}[\ketbra{\Psi}{\Psi}_{AQ}].
\end{equation}
On top of this constraint, we have the constraint on the normalisation of the quantum state $\rho_{AB}$
\begin{equation}
    \Tr[\rho_{AB}] = 1,
\end{equation}
the positive-semidefiniteness of the state
\begin{equation}
    \rho_{AB} \succeq 0,
\end{equation}
and the constraints due to the statistics that are being monitored in the test rounds of the protocol
\begin{equation}
    \Tr[\rho_{AB} \Pi_c] = q_c, \qquad \forall c \in \cC \setminus \{\perp\}.
\end{equation}
where $\Pi_c$ is the POVM element associated to the score $C = c$.

Putting all these things together, we have the following optimisation problem
\begin{equation} \label{eq: annex original optimisation}
    \begin{split}
        g(q) = \inf_{\rho_{AB}} \quad & H(Z|L,E)_{\cN[\rho_{AB}]}\\
        \text{s.t.} \quad & \rho_{AB} \succeq 0,\\
        &\Tr[\rho_{AB}] = 1,\\
        &\Tr_{B}[\rho_{AB}] = \Tr_Q[\ketbra{\Psi}{\Psi}_{AQ}]\\
        & \Tr[\rho_{AB} \Pi_c] = q_c, \qquad \forall c \in \cC \setminus \{\perp\}.
    \end{split}
\end{equation}
Here, $\cN$ is the quantum channel that describes Bob's measurement (which converts the register $B$ to $Z$), the announcement of register $L$ (which depends on $T$ and $Z$), and lastly, tracing out Alice's register $A$. Without loss of generality, we assume that the quantum side-information $E$ that is available to Eve is the purification of the state $\rho_{AB}$. As mentioned earlier, the above optimisation involves an optimisation over infinite-dimensional quantum state $\rho_{AB}$ since $\cH_B$ is infinite-dimensional. Our strategy is to truncate the Hilbert space dimension by projecting $\rho_{AB}$ into a finite-dimensional subspace of the full Hilbert space, and calculate the correction terms due to this truncation. The rest of Appendix~\ref{annex: dimension reduction} is dedicated to these calculations.

\subsection{Bounding the trace distance between truncated and original state}
The most important ingredient to obtain the correction terms due to the Hilbert truncation is an upper bound on the trace distance between the truncated state and the original state. We let $\rho_{ABE}$ be the state shared between Alice, Bob and Eve. Without loss of generality, we can assume $\rho_{ABE}$ to be a pure state: there exists a normalised vector $\ket{\psi}_{ABE}$ such that $\rho_{ABE} = \ketbra{\psi}{\psi}_{ABE}$. Just like in the main text, we fix $\nmax \in \mathbb{N}$ and denote
\begin{equation}
    N_0 = \sum_{n = 0}^{\nmax} \ketbra{n}{n}
\end{equation}
as the projection to the ``low energy subspace'' . Let $N_1 = \1 - N_0$. Our goal is to find some small $\delta$ such that
\begin{equation}
    \Delta(\rho_{ABE}, (\1_A \otimes N_0 \otimes \1_E) \rho_{ABE} (\1_A \otimes N_0 \otimes \1_E)) \leq \delta,
\end{equation}
where $\Delta$ denotes the trace distance.

For a fix $\wmax$, let
\begin{equation} \label{eq: large heterodyne outcome}
    V_1 = \int_{\abs{\beta}^2 > \wmax} \dd^2 \beta \frac{\ketbra{\beta}{\beta}}{\pi}
\end{equation}
be the POVM element corresponding to the ``large heterodyne outcome''. From Appendix A Eq. (A3) of Ref.~\cite{kanitschar2023finite}, the following operator inequality holds.
\begin{equation} \label{eq: operator inequality}
    N_1 \preceq \frac{\Gamma(\nmax + 2, 0)}{\Gamma(\nmax + 2,\wmax)} V_1.
\end{equation}
Here $\Gamma$ denotes the upper incomplete Gamma function.
\begin{equation}
    \Gamma(n, x) = \int_{x}^{\infty} \dd t \, t^{n-1} e^{-t}.
\end{equation}
Consequently, for any state $\sigma$, we have
\begin{equation}
    \Tr[\sigma N_1] \leq \frac{\Gamma(\nmax+2, 0)}{\Gamma(\nmax+2,\wmax)} \Tr[\sigma V_1]
\end{equation}
In passing, we note that there is a slight difference in how we define $N_0$ and $N_1$ as compared to the corresponding quantities in Ref.~\cite{kanitschar2023finite}. In Ref.~\cite{kanitschar2023finite}, the cut-off photon number is defined as $n_c = \nmax + 1$.

Now, the trace distance can be upper bounded in terms of the weight of the state in the ``high energy space'' using the gentle measurement lemma (Lemma 9 of Ref.~\cite{winter1999coding})
\begin{equation*}
\begin{split}
    \Delta(\rho_{ABE}, (\1_A \otimes N_0 \otimes \1_E) \rho_{ABE} (\1_A \otimes N_0 \otimes \1_E))
    &\leq \sqrt{2 \Tr[\rho_{ABE} (\1_A \otimes N_1 \otimes \1_E)]}\\
    &\leq \sqrt{2\frac{\Gamma(\nmax+2, 0)}{\Gamma(\nmax +2,\wmax)}\Tr[\rho_B V_1]} \\
    &= \sqrt{2\frac{\Gamma(\nmax+2, 0)}{\Gamma(\nmax+2,\wmax)}\Pr[W > \wmax]}
\end{split}
\end{equation*}
where the last term can be monitored in the test rounds as $\freq(C = \top)$. Let us introduce the shorthand $\Pr[W > \wmax] = \nu$, and $\kappa = \Gamma(\nmax + 2, 0)/\Gamma(\nmax + 2,\wmax)$. We can write $\delta = \sqrt{2\kappa \nu}$.

The gentle measurement lemma given in Ref.~\cite{winter1999coding} applies to any state and measurement. However, for our case, we can tighten the bound as we are dealing with a special case of gentle measurement. First, notice that $\ket{\psi}_{ABE}$ can be further decomposed as
\begin{equation}
    \ket{\psi}_{ABE} = \sqrt{1-w} \ket{\psi_0} + \sqrt{w} \ket{\psi_1},
\end{equation}
where $w = \bra{\psi} (\1 \otimes N_1 \otimes \1) \ket{\psi}$ is the weight of the and
\begin{equation}
    \begin{split}
        (\1 \otimes N_0 \otimes \1) \ketbra{\psi}{\psi}_{ABE} (\1 \otimes N_0 \otimes \1) &= (1-w) \ketbra{\psi_0}{\psi_0},\\
         (\1 \otimes N_1 \otimes \1) \ketbra{\psi}{\psi}_{ABE} (\1 \otimes N_1 \otimes \1) &= w \ketbra{\psi_1}{\psi_1}.
    \end{split}
\end{equation}
Since $N_0$ and $N_1$ are orthogonal, $\ket{\psi_0}$ and $\ket{\psi_1}$ are also orthogonal. They are also normalised. Then, the trace distance between the truncated state and the original state is given by
\begin{equation}
    \Delta\left(\ketbra{\psi}{\psi}_{ABE}, (1-w) \ketbra{\psi_0}{\psi_0}\right) = \frac{1}{2} \Tr \left[\sqrt{\left(\ketbra{\psi}{\psi} - (1-w) \ketbra{\psi_0}{\psi_0} \right)^\dagger \left(\ketbra{\psi}{\psi} - (1-w) \ketbra{\psi_0}{\psi_0} \right)}\right]
\end{equation}
To compute the trace distance, we first compute
\begin{equation}
    M = \ketbra{\psi}{\psi} - (1-w) \ketbra{\psi_0}{\psi_0} =
    \begin{pmatrix}
        0 & \sqrt{w(1-w)}\\
        \sqrt{w(1-w)} & w
    \end{pmatrix},
\end{equation}
then $M^\dagger M$ is given by
\begin{equation}
    M^\dagger M =
    \begin{pmatrix}
        w(1-w) & w\sqrt{w(1-w)} \\
        w\sqrt{w(1-w)} & w^2 + w(1-w)
    \end{pmatrix}
\end{equation}
The eigenvalues of $M^\dagger M$ are given by
\begin{equation}
    \mathrm{eig}_{\pm}[M^\dagger M] =  \frac{w(2-w) \pm w\sqrt{w(4-3w)}}{2}
\end{equation}
Therefore, the trace distance is given by
\begin{align}
    &\Delta\left(\ketbra{\psi}{\psi}_{ABE}, (1-w) \ketbra{\psi_0}{\psi_0}\right) = \Tr[\sqrt{M^\dagger M}] = \sum_{j = +,-} \sqrt{\mathrm{eig}_j[M^\dagger M]}\\
    &= \frac{1}{2\sqrt{2}} \left[\sqrt{w(2-w) + w\sqrt{w(4-3w)}} + \sqrt{w(2-w) - w\sqrt{w(4-3w)}} \right] =: \delta(w)
\end{align}
The behaviour of the trace distance $\delta$ against the weight for the high energy subspace $w$ is given in Fig.~\ref{fig:delta vs w}.
\begin{figure}[h]
    \centering
    \includegraphics{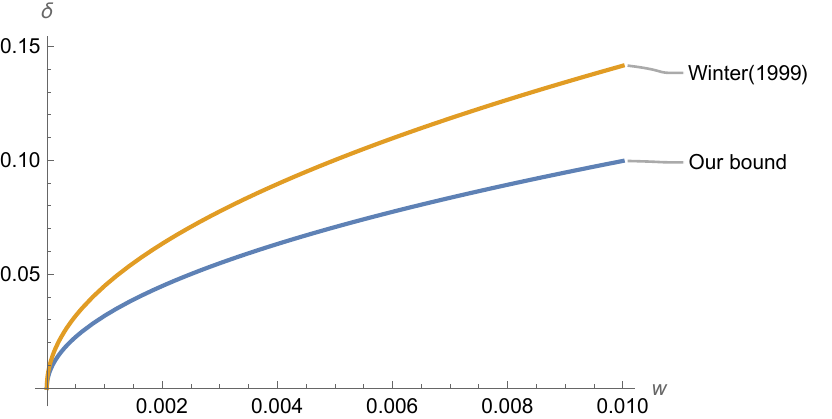}
    \caption{The comparison of the behaviour of $\delta$ vs $w$ for $w \in (0,0.01)$ between our bound and the one in Ref.~\cite{winter1999coding}.}
    \label{fig:delta vs w}
\end{figure}

However, we cannot access the value of $w$ directly since we only have an upper bound $w \leq \kappa \nu$. Since our bound is not monotonously non-decreasing, we need to convert it to a bound that have this monotonicity property. First, we take the first derivative
\begin{equation}
    \frac{\dd \delta}{\dd w} = \frac{(1-w)}{2\sqrt{2}w\sqrt{4-3w}} \left[\frac{3w + \sqrt{w(4-3w)}}{\sqrt{(2-w) + \sqrt{w(4-3w)}}} - \frac{3w - \sqrt{w(4-3w)}}{\sqrt{(2-w) - \sqrt{w(4-3w)}}} \right]
\end{equation}
The pre-factor is non-negative for $0 \leq w \leq 1$. We want the term in the square bracket to be non-negative. From Mathematica, this condition is satisfied if $w \leq 2/3$. Thus, the maximum trace distance is given by $\delta(\frac{2}{3}) = \frac{1}{\sqrt{3}}$. Therefore, as a function of $\nu$, we have
\begin{equation}
    \delta = \begin{cases}
        \frac{1}{2\sqrt{2}} \left[\sqrt{\kappa \nu (2-\kappa \nu) + \kappa \nu \sqrt{\kappa \nu(4-3\kappa \nu)}} + \sqrt{\kappa \nu(2-\kappa \nu) - \kappa \nu \sqrt{\kappa \nu(4-3\kappa \nu)}} \right] & \text{if $\kappa \nu \leq 2/3$}\\
        \frac{1}{\sqrt{3}} & \text{otherwise}.
    \end{cases}
\end{equation}

To construct the min-tradeoff function, we want to linearise the upper bound on the trace distance. First, we check that the function $\delta$ is concave in $w$, so that we can simply take a tangent line. We take the second derivative
\begin{equation}
    \frac{\dd^2 \delta}{\dd w^2} = - \frac{(\sqrt{w} + \sqrt{4-3w})\sqrt{(2-w) - \sqrt{w(4-3w)}} - (\sqrt{w} - \sqrt{4-3w})\sqrt{(2-w) + \sqrt{w(4-3w)}}}{2\sqrt{2}(1-w) [w(4-3w)]^{3/2}}
\end{equation}
Since $\sqrt{w} < \sqrt{4-3w}$ for any $0 \leq w \leq 1$, both the numerator and denominator (excluding the negative sign outside) are positive. Thus, we have $\dd^2 \delta / \dd w^2 < 0$ which means that $\delta$ is concave in the interval $0 \leq w \leq 1$. Since $w = \kappa \nu$ is linear in $\nu$, the trace distance $\delta$ is also concave in $\nu$. 

Therefore, we can simply take a tangent line at any $\nu$ such that $\kappa \nu \leq 2/3$ to derive a linear upper bound on $\delta$. Suppose we take the tangent at $\nu = \nu_0$ with $\nu_0 \in (0, \frac{2}{3\kappa})$, we use the shorthand $w_0 = \kappa \nu_0$
\begin{multline} \label{eq: delta bound in terms of w}
    \delta \leq \frac{(1-w_0)}{2\sqrt{2}\sqrt{4-3w_0}} \left[\frac{3w_0 + \sqrt{w_0(4-3w_0)}}{\sqrt{(2-w_0) + \sqrt{w_0(4-3w_0)}}} - \frac{3w_0 - \sqrt{w_0(4-3w_0)}}{\sqrt{(2-w_0) - \sqrt{w_0(4-3w_0)}}} \right] \frac{(\nu - \nu_0)}{\nu_0}\\
    + \frac{1}{2\sqrt{2}} \left[\sqrt{w_0(2-w_0) + w_0\sqrt{w_0(4-3w_0)}} + \sqrt{w_0(2-w_0) - w_0\sqrt{w_0(4-3w_0)}} \right].
\end{multline}
Therefore,
\begin{equation} \label{eq: annex delta m0 c0}
    \delta \leq m_0 (\nu - \nu_0) + c_0,
\end{equation}
with
\begin{align}
    m_0 &=\frac{(1-w_0)}{2\sqrt{2} \nu_0 \sqrt{4-3w_0}} \left[\frac{3w_0 + \sqrt{w_0(4-3w_0)}}{\sqrt{(2-w_0) + \sqrt{w_0(4-3w_0)}}} - \frac{3w_0 - \sqrt{w_0(4-3w_0)}}{\sqrt{(2-w_0) - \sqrt{w_0(4-3w_0)}}} \right], \label{eq: annex m0}\\
    c_0 &= \frac{1}{2\sqrt{2}} \left[\sqrt{w_0(2-w_0) + w_0\sqrt{w_0(4-3w_0)}} + \sqrt{w_0(2-w_0) - w_0\sqrt{w_0(4-3w_0)}} \right] \label{eq: annex c0}.
\end{align}

\subsection{Continuity of conditional von Neumann entropy}
Now that we have an upper bound on the trace distance between the truncated and the original state, we can use the continuity bound for the von Neumann entropy to prove that the truncation does not significantly reduce the entropy. We use the following Lemma 1 from Ref.~\cite{upadhyaya2021dimension}
\begin{lemma}[Continuity of conditional von Neumann entropy] \label{lemma: continuity entropy}
    Let $\cH_Z$ and $\cH_E$ be two Hilbert spaces, where $\mathrm{dim}(\cH_Z) = d_Z$ while $\cH_E$ can be infinite dimensional. Let $\rho, \sigma \in \mathsf{D}(\cH_Z \otimes \cH_E)$ be subnormalised states with $\Tr[\rho] \geq \Tr[\sigma]$. If $\Delta(\rho, \sigma) \leq \delta$, then
    \begin{equation}
        H(Z|E)_\rho \geq H(Z|E)_\sigma - \delta \log_2 d_Z - (1 + \delta) h_2\left(\frac{\delta}{1+\delta}\right),
    \end{equation}
    where $h_2(x) = -x \log_2(x) - (1-x) \log_2(1-x)$ is the binary entropy function.
\end{lemma}
Now, let $\cN: ABE \rightarrow ZE$ be the quantum channel that corresponds to Bob's measurement (and tracing out Alice). In this context, the state $\rho$ in the above lemma corresponds to $\rho = \cN[\rho_{ABE}]$ and $\sigma = \cN[(\1_A \otimes N_0 \otimes \1_E) \rho_{ABE} (\1_A \otimes N_0 \otimes \1_E)]$. Because of the monotonicity property of trace distance, we can also re-use the upper bound on the trace distance that we have derived in the previous subsection.

The correction term for the conditional entropy is given by
\begin{equation}
    v(\delta) = \delta \log_2 d_Z + (1 + \delta) h_2 \left(\frac{\delta}{1 + \delta}\right)
\end{equation}
Again, for the purpose of constructing min-tradeoff functions, we want to obtain a bound that is affine in $\nu$. We adopt a similar strategy of taking an upper bound by taking a tangent line at $\nu = \nu_c$ where $\nu_c \in (0, 2/3\kappa)$. We need to calculate the gradient of $v$ against $\nu$. To that end, we calculate the gradient of $v$ against $\delta$ 
\begin{equation}
    \frac{\dd v}{\dd \delta} = \log_2 d_Z + \log_2 \left(\frac{1 + \delta}{\delta} \right).
\end{equation}
On the other hand, from Eq.~\eqref{eq: delta bound in terms of w}, the gradient of $\delta$ against $\nu$ is given by
\begin{equation}
    \frac{\dd \delta}{\dd \nu} = \frac{\dd \delta}{\dd w} \frac{\dd w}{\dd \nu} = \frac{(1-w)}{2\sqrt{2}w\sqrt{4-3w}} \left[\frac{3w + \sqrt{w(4-3w)}}{\sqrt{(2-w) + \sqrt{w(4-3w)}}} - \frac{3w - \sqrt{w(4-3w)}}{\sqrt{(2-w) - \sqrt{w(4-3w)}}} \right] \cdot \kappa,
\end{equation}
where $w = \kappa \nu$ depends implicitly on $\nu$. We introduce the shorthand 
\begin{align}
w_c &= \kappa \nu_c\\
\delta_c &= \frac{1}{2\sqrt{2}}\left[\sqrt{w_c (2 - w_c) + w_c\sqrt{w_c (4 - 3 w_c)}} + \sqrt{w_c (2 - w_c) - w_c \sqrt{w_c (4 - 3 w_c)}} \right]
\end{align}
Then,
\begin{align}
    m_{\corr} = \frac{\dd v}{\dd \nu} \Bigg\vert_{\nu = \nu_c} &= \left(\frac{\dd v}{\dd \delta} \cdot \frac{\dd \delta}{\dd \nu} \right) \Bigg\vert_{\nu = \nu_c}\nonumber\\
    &= \frac{(1-w_c) \kappa}{2\sqrt{2}w_c \sqrt{4 - 3w_c}}  \left[\log_2 d_Z + \log_2 \left(\frac{1 + \delta_c}{\delta_c} \right)\right] \nonumber \\
    & \qquad \times \left[\frac{3 w_c + \sqrt{w_c (4 - 3 w_c)}}{\sqrt{(2 - w_c) + \sqrt{w_c (4 - 3 w_c)}}} - \frac{3w_c - \sqrt{w_c (4 - 3w_c)}}{\sqrt{(2 - w_c) - \sqrt{w_c (4 - 3w_c)}}} \right] 
\end{align}
and
\begin{equation}
    c_{\corr} = v(\delta_c) = \delta_c \log_2 d_Z + (1 + \delta_c) h_2 \left(\frac{\delta_c}{1 + \delta_c}\right).
\end{equation}
This gives a linear correction term to the conditional entropy 
\begin{equation}
    g^{(\nu_c)}_{\corr}(\nu) = m_{\corr}(\nu - \nu_c) + c_{\corr}.
\end{equation}

\subsection{Correction to the constraints}
The projection to the low energy subspace does not only affect the conditional entropy but also the constraints. We consider the entanglement based picture before Bob's measurement and denote the entangled state between Alice and Bob by $\rho_{AB}$ and $\sigma_{AB} = (\1 \otimes N_0) \rho_{AB} (\1 \otimes N_0)$.

\subsubsection*{Normalisation constraint}
The first constraint is the normalisation, this modification is straightforward, we have
\begin{equation}
    1 - \kappa \nu \leq \Tr[\sigma_{AB}] \leq 1.
\end{equation}

\subsubsection*{Constraint on Alice's marginal state}
Next, we have the constraint on Alice's marginal states. Originally, we have the constraint
\begin{equation*}
    \rho_A = \Tr_Q \left[\ketbra{\Psi}{\Psi}_{AQ}\right].
\end{equation*}
However, due to the cut-off, we instead have to constrain the state
\begin{equation*}
    \sigma_A=(1-w)\Tr_{BE}[\ketbra{\psi_0}{\psi_0}_{ABE}].
\end{equation*}
Let us expand the original state $\rho_A$, 
\begin{align*}
    \rho_A &= \Tr_{BE}[\ketbra{\psi}{\psi}_{ABE}]\\
    &= w\Tr_{BE}[\ketbra{\psi_1}{\psi_1}_{ABE}] + \sqrt{w(1-w)}\Tr_{BE}[\ket{\psi_0}\bra{\psi_1}_{ABE}] + \sqrt{w(1-w)}\Tr_{BE}[\ket{\psi_1}\bra{\psi_0}_{ABE}] +  \sigma_A
\end{align*}
By expanding $\ket{\psi}_{ABE}=\sum_ac_a\ket{a}_A\otimes\ket{\phi_a}_{BE}$, the cross-terms can be shown to be
\begin{align*}
    \Tr_{BE}[\ket{\psi_0}\bra{\psi_1}_{ABE}] &= \Tr_{BE}[(\1 \otimes N_0 \otimes \1)\ket{\psi}\bra{\psi}_{ABE}(\1 \otimes N_1 \otimes \1)]\\
    &= \sum_{aa'}c_a(c_a')^*\Tr[(N_0 \otimes \1)\ket{\phi_a}\bra{\phi_a}_{BE}(N_1 \otimes \1)]\\
    &= 0,
\end{align*}
where the final equality notes that the projectors $N_0$ and $N_1$ are orthogonal.
As such, we have the following constraint,
\begin{align*}
    \Delta(\sigma_A,\Tr_Q \ketbra{\Psi}{\Psi}_{AQ}) &= \Delta(\sigma_A,\rho_A)\\
    &= \frac{1}{2} \norm{w\Tr_{BE}[\ketbra{\psi_1}{\psi_1}_{ABE}]}_1\\
    &\leq \frac{1}{2}w,
\end{align*}
where the second line invokes the definition of trace distance, and the third line utilise the fact that the trace norm of a quantum state is bounded by 1.

The trace distance constraint can be re-cast as SDP constraints (Example 1.20 of Ref.~\cite{watrous2018theory}) by defining the matrix $M$
\begin{equation}
    M = \begin{pmatrix}
        \zeta_1 & \sigma_A - \Tr_Q \ketbra{\Psi}{\Psi} \\
        \sigma_A - \Tr_Q \ketbra{\Psi}{\Psi} & \zeta_2
    \end{pmatrix}
\end{equation}
then imposing that $\Tr[M] \leq \kappa\nu$  and $M \succeq 0$.



\subsubsection*{Statistical constraints}
Finally, the most challenging correction is to the statistical constraints. To do this,
we can write the original state $\rho$ as 
\begin{equation} \label{eq: rho blocks}
    \rho = \begin{pmatrix}
        \rho_0 & \rho_c\\
        \rho_c^{\dagger} & \rho_h
    \end{pmatrix}.
\end{equation}
Similarly, we can write the measurement operator as
\begin{equation} \label{eq: Pi blocks}
    \Pi = \begin{pmatrix}
        \Pi_0 & \Lambda \\ 
        \Lambda^\dagger & \Pi_h
    \end{pmatrix}.
\end{equation}
Let $\rho_0$ and $\Pi_0$ be linear operators in the low-energy subspace $\cH_{\text{low}}$. We shall truncate the subspace outside of $\cH_{\text{low}}$. The constraint involves terms of the form
\begin{align*}
    \Tr[\rho \Pi] &= \Tr[\rho_0 \Pi_0] + \Tr[\rho_h \Pi_h] + \Tr[\rho_c \Lambda^\dagger] + \Tr[\rho_c^\dagger \Lambda] \\
    &= \Tr[\rho_0 \Pi_0] + \Tr[\rho_h \Pi_h] + 2 \Tr[\rho_c \Lambda^\dagger].
\end{align*}
The LHS is simply the original constraint which is related to the statistics that we see in the experiment. The second inequality is due to the fact that $\Tr[X] = \Tr[X^\dagger]$ and the cyclic permutation symmetry of trace. Bounding the $\Tr[\rho_h \Pi_h]$ is easy:
\begin{equation}
    0 \leq \Tr[\rho_h \Pi_h] \leq \kappa \nu,
\end{equation}
since $0 \preceq \Pi_h \preceq \1$. Next, we need to bound $\Tr[\rho_c \Lambda^\dagger]$. To do that we need the following lemmas
\begin{lemma}[Schur complement of positive semidefinite matrices] \label{lemma: Schur complement}
Let $A \succeq 0$, $B \succeq 0$ and
\begin{equation*}
    \begin{pmatrix}
        A & X \\
        X^\dagger & B
    \end{pmatrix} \succeq 0
\end{equation*}
be positive semidefinite. Then, we have
\begin{equation*}
    B - X^\dagger A^{-1} X \succeq 0
\end{equation*}
\end{lemma}

\begin{corollary} \label{cor: contraction lemma}
    Let $A \succeq 0$, $B \succeq 0$ and
\begin{equation*}
    \begin{pmatrix}
        A & X \\
        X^\dagger & B
    \end{pmatrix} \succeq 0
\end{equation*}
be positive semidefinite. Then, there exists some matrix $K$ such that
\begin{equation*}
    X = A^{1/2} K B^{1/2}
\end{equation*}
and
\begin{equation*}
    K^\dagger K \preceq \1.
\end{equation*}
\end{corollary}
\begin{proof}
Based on Lemma~\ref{lemma: Schur complement}, we have
\begin{equation*}
    B \succeq X^\dagger A^{-1} X.    
\end{equation*}
Hence, by multiplying both sides with $B^{-1/2}$ to the left and to the right, we have
\begin{equation*}
    \1 \succeq B^{-1/2} X^\dagger A^{-1} X B^{-1/2} = \left( A^{-1/2} X B^{-1/2}\right)^\dagger \left( A^{-1/2} X B^{-1/2}\right)
\end{equation*}
We define
\begin{equation*}
    K := A^{-1/2} X B^{-1/2}.
\end{equation*}
It is straightforward to see that this implies the two claims.
\end{proof}

\begin{lemma}[Cauchy-Schwarz inequality for trace] \label{lemma: CS inequality}
Let $A$ and $B$ be matrices. Then,
\begin{equation*}
    \Tr[A^\dagger B]^2 \leq \Tr[A^\dagger A] \cdot \Tr[B^\dagger B]
\end{equation*}
\end{lemma}

\begin{lemma}[Bounds on the cross terms]
Let $\rho$ and $\Pi$ be given by Eqs.~\eqref{eq: rho blocks} and \eqref{eq: Pi blocks}, respectively. Then, we have
\begin{equation}
    -\sqrt{\Tr[\rho_0 \Pi_0] \Tr[\rho_h \Pi_h]} \leq \Tr[\rho_c \Lambda^\dagger] \leq \sqrt{\Tr[\rho_0 \Pi_0] \Tr[\rho_h \Pi_h]}
\end{equation}
\end{lemma}
\begin{proof}
Based on Corollary~\ref{cor: contraction lemma}, we can write
\begin{align*}
    \rho_c &= \rho_0^{1/2} K_1 \rho_h^{1/2}\\
    \Lambda &= \Pi_0^{1/2} K_2 \Pi_h^{1/2},
\end{align*}
for some $K_1, K_2$ such that $K_1^\dagger K_1 \preceq \1$ and $K_2^\dagger K_2 \preceq \1$.

Then, we can write
\begin{align*}
    \Tr[\rho_c \Lambda^\dagger] &= \Tr\left[\left(\rho_0^{1/2} K_1 \rho_h^{1/2} \right)\cdot \left(\Pi_h^{1/2} K_2^\dagger \Pi_0^{1/2}\right) \right] \\
    &= \Tr\left[ \left(\Pi_0^{1/2} \rho_0^{1/2}\right) \cdot \left(K_1 \rho_h^{1/2} \Pi_h^{1/2} K_2^\dagger \right) \right]\\
    &\leq \sqrt{\Tr\left[ \Pi_0 \rho_0 \right] \cdot \Tr\left[\left(K_1 \rho_h^{1/2} \Pi_h^{1/2} K_2^\dagger \right) \cdot \left(K_2 \Pi_h^{1/2} \rho_h^{1/2} K_1^\dagger \right)\right]}\\
    &\leq \sqrt{\Tr\left[ \Pi_0 \rho_0 \right] \cdot \Tr\left[\rho_h\Pi_h\right]}
\end{align*}
where in the first line, we substitute $\rho_c$ and $\Lambda^\dagger$ with their identities which involved $K_1$ and $K_2$. In the second line, we use the cyclic permutation symmetry of trace and group the terms accordingly. In the third line, we apply Lemma~\ref{lemma: CS inequality}. In the fourth line, we use the fact that $K_1^\dagger K_1 \preceq \1$, $K_2^\dagger K_2 \preceq \1$ and the monotonicity of trace. This proves the upper bound claim. The lower bound claim can be obtained similarly by modifying the third line where we consider the negative square root of Lemma~\ref{lemma: CS inequality} instead.
\end{proof}

\begin{corollary}\label{cor: correction to statistics constraints}
    Let $\rho$ and $\Pi$ be given by Eqs.~\eqref{eq: rho blocks} and \eqref{eq: Pi blocks}, respectively. Let $p = \Tr[\rho \Pi]$ and $w \geq \Tr[\rho_h]$. Then, we have
    \begin{equation}
        p - \xi_L \leq \Tr[\rho_0 \Pi_0] \leq p - \xi_U
    \end{equation}
    where
    \begin{equation}
    \begin{split}
        \xi_L &= \begin{cases}
            w + 2\sqrt{w(1-w)} & \mathrm{if} \quad w < \frac{5+ \sqrt{5}}{10}\\
            \frac{1+\sqrt{5}}{2} & \mathrm{if} \quad w \geq \frac{5+ \sqrt{5}}{10}     
       \end{cases}\\
        \xi_U &= \begin{cases}
            w - 2 \sqrt{w(1-w)}  &\mathrm{if} \quad w < \frac{5 - \sqrt{5}}{10}\\
            \frac{1 - \sqrt{5}}{2} \qquad &\mathrm{if} \quad w \geq \frac{5 - \sqrt{5}}{10}
        \end{cases}
    \end{split}
    \end{equation}
\end{corollary}
\begin{proof}
First, we write
\begin{equation*}
    \Tr[\rho_0 \Pi_0] = \Tr[\rho \Pi] - \Tr[\rho_h \Pi_h] - 2 \Tr[\rho_c \Lambda^\dagger]
\end{equation*}
Thus, the correction term $\xi$ is given by
\begin{equation*}
    \xi = \Tr[\rho_h \Pi_h] + 2 \Tr[\rho_c \Lambda^\dagger]
\end{equation*}
We assume that $\Tr[\rho_h] = \omega \leq w$. We remark that $\omega$ is not observable in practice and we can only bound it with $w$ which is a function of $\nu$.

To prove the lower bound, this is easy because we simply have to maximise the correction term
\begin{equation*}
     \xi \leq \Tr[\rho_h \Pi_h] + 2 \sqrt{\Tr[\rho_h \Pi_h]\cdot \Tr[\rho_0 \Pi_0]} \leq \omega + 2 \sqrt{\omega(1-\omega)}
\end{equation*}
Differentiating the upper bound with respect to $\omega$, we see that there is a critical $\omega$ at $\omega_c = (5 + \sqrt{5})/10 \approx 0.724$. Therefore, we have
\begin{equation}
    \xi \leq \xi_L := \begin{cases}
        w + 2\sqrt{w(1-w)} &\text{if $w < \omega_c$}\\
        \frac{1 + \sqrt{5}}{2} & \text{if $w \geq \omega_c$}
    \end{cases}
\end{equation}

The proof for the upper bound is slightly more complicated.  We now have to minimise the correction term $\xi$.
\begin{equation*}
    \xi \geq \Tr[\rho_h \Pi_h] - 2 \sqrt{\Tr[\rho_h \Pi_h] \cdot \Tr[\rho_0 \Pi_0]} = \sqrt{\Tr[\rho_h \Pi_h]} \cdot \left(\sqrt{\Tr[\rho_h \Pi_h]} - 2 \sqrt{\Tr[\rho_0 \Pi_0]} \right)
\end{equation*}
The minimisation over $\Pi_0$ is easy: we want $\rho_0$ and $\Pi_0$ to be perfectly overlapped and hence $\Tr[\rho_0 \Pi_0] \geq 1 - \omega$. We then parameterise $\Tr[\rho_h \Pi_h] = \omega \lambda$ for some $\lambda \in [0,1]$. Thus, we have
\begin{equation*}
    \xi \geq \sqrt{\omega \lambda} \cdot \left(\sqrt{\omega \lambda} - 2 \sqrt{1 - \omega} \right) := F(\omega, \lambda)
\end{equation*}
Differentiating with respect to $\lambda$, we have
\begin{equation*}
    \frac{\partial F}{\partial \lambda} = \omega \left(1- \sqrt{\frac{1-\omega}{\omega \lambda}}\right) 
\end{equation*}
For $\omega < 1/2$, we have $\frac{\partial F}{\partial \lambda} < 0$. On the other hand, for $\omega \geq 1/2$, we have $\frac{\partial F}{\partial \lambda} = 0$ when $\lambda = (1-\omega)/\omega$. We consider the case where $\omega < 1/2$ first. In this case, it is optimal to take $\lambda = 1$ since the first derivative with respect to $\lambda$ is negative. In the second case, we shall take $\lambda = (1-\omega)/\omega$. Therefore, our correction term is given by
\begin{equation*}
    \xi \geq G(\omega) :=\begin{cases}
    \omega - 2 \sqrt{\omega(1-\omega)} & \text{if $\omega < 1/2$}\\
    \omega - 1 & \text{if $\omega \geq 1/2$}
    \end{cases}
\end{equation*}
Note that $G(\omega) \leq 0$ for all $\omega \in [0,1]$, hence $p - G(\omega) \geq p$. Taking the first derivative of $G$ with respect to $\omega$:
\begin{equation*}
    \frac{\mathrm{d} G}{\mathrm{d} \omega} =
    \begin{cases}
        \frac{2\omega + \sqrt{\omega(1-\omega)}-1}{\sqrt{\omega(1-\omega)}} & \text{for $\omega < 1/2$}\\
        1 & \text{for $\omega \geq 1/2$}
    \end{cases}.
\end{equation*}
Thus, it is clear to see that the optimal $\omega$ must be within the [0, 1/2) interval. We solve for $\frac{\mathrm{d}G}{\mathrm{d}\omega} = 0$ in this interval and obtain the optimal $\omega_* = (5-\sqrt{5})/10 \approx 0.276$. Also, we have $\frac{\mathrm{d}G}{\mathrm{d}\omega} < 0$ for $\omega < \omega_*$ and $\frac{\mathrm{d}G}{\mathrm{d}\omega} > 0$ for $\omega > \omega_*$, which implies the minimum $\xi$ is given by
\begin{equation*}
    \xi \geq \xi_+ := \begin{cases}
        w - 2\sqrt{w(1-w)} &\text{if $w < \omega_*$}\\
        \frac{1-\sqrt{5}}{2} &\text{if $w \geq \omega_*$}
    \end{cases}.
\end{equation*}
This concludes the proof of the upper bound, and hence the claim.
\end{proof}

Corollary~\ref{cor: correction to statistics constraints} gives an upper and lower bound on the statistical constraints in terms of the actual probability $p$ associated to the original state and the upper bound on the weight of the high photon number components $w$. However, these bounds are not linear in $w$ (hence, not linear in $\nu$) whereas for the purpose of constructing min-tradeoff functions, it is convenient to obtain correction terms that are linear in $\nu$. This ensures that the dual solution obtained from our numerical methods can be readily used as min-tradeoff functions \footnote{Alternatively, we can keep the expressions for $\xi_L$ and $\xi_U$ as they are and linearise the min-tradeoff functions later.}

Fortunately, the correction terms $\xi_L$ (and $\xi_U$) can be easily linearised since they are concave (and convex, respectively). By taking a tangent line at any point on the curve, we will obtain an upper bound on $\xi_L$ (and lower bound on $\xi_U$), which relaxes the constraint in the right direction. We denote these bounds by $\hat{\xi}_L^{(\nu_L)}$ and $\hat{\xi}_U^{(\nu_U)}$, respectively. Here $\nu_L$ and $\nu_U$ denotes the point in which the tangent line is taken.

Let $\nu_L \in (0,(5+\sqrt{5})/10\kappa]$. Then, $\hat{\xi}_L^{(\nu_L)}(\nu)$ is given by 
\begin{equation} \label{eqn: linear xi_L}
    \hat{\xi}_L^{(\nu_L)}(\nu) = \left(1 + \frac{(1- 2 \kappa \nu_L)}{\sqrt{\kappa \nu_L (1 - \kappa \nu_L)}} \right) \kappa\nu + \sqrt{\frac{\kappa \nu_L}{1-\kappa \nu_L}}
\end{equation}
On the other hand, let $\nu_U \in (0, (5-\sqrt{5})/10\kappa]$. Then, $\hat{\xi}_U^{(\nu_U)}(\nu)$ is given by
\begin{equation} \label{eqn: linear xi_U}
    \hat{\xi}_U^{(\nu_U)}(\nu) = \left(1 - \frac{(1- 2 \kappa \nu_U)}{\sqrt{\kappa \nu_U (1 - \kappa \nu_U)}} \right) \kappa\nu - \sqrt{\frac{\kappa \nu_U}{1-\kappa \nu_U}}
\end{equation}

\subsection{Proof of dimension reduction theorem}
Putting everything together, we can define the function $\tilde{g}'(q)$
\begin{equation} \label{eq: annex truncated optimisation}
\begin{split}
    \tilde{g}'(q) = \inf_{\sigma_{AB}} \quad & H(Z|L,E)_{\cN[\sigma_{AB}]}\\
    \text{s.t.} \quad & \sigma_{AB} \succeq 0,\\
    & \Tr[\sigma_{AB}] \leq 1,\\
    & \Tr[\sigma_{AB}] \geq 1 - \kappa q_{\top},\\
    & \Delta(\sigma_A, \Tr_Q[\ketbra{\Psi}{\Psi}_{AQ}]) \leq \kappa q_{\top}/2,\\
    & \Tr[\sigma_{AB} \Pi_c] \leq q_c - \xi_U(q_{\top}), \\
    & \Tr[\sigma_{AB} \Pi_c] \geq q_c + \xi_L(q_{\top})
\end{split}
\end{equation}
Here, we use the notation $\sigma_{AB} = (\1_A \otimes N_0) \rho_{AB} (\1_A \otimes N_0)$. We can define the function $\tilde{g}(q)$ as the linearised version of the above function which can be obtained by replacing $\xi_L$ and $\xi_U$ with $\hat{\xi}_L$ and $\hat{\xi}_U$, respectively. (Refer to Eqs.~\eqref{eqn: linear xi_L} and \eqref{eqn: linear xi_U} for those expressions). Then, using Lemma~\ref{lemma: continuity entropy}, we have
\begin{align}
    g(q) &= \inf_{\rho_{AB} \in \cS_{q}} H(Z|L,E)_{\cN[\rho_{AB}]} \nonumber\\
    &\geq \inf_{\sigma_{AB}: \, \Delta(\rho_{AB}, \sigma_{AB}) \leq \delta(q_{\top}), \, \rho_{AB} \in \cS_{q}} H(Z|L,E)_{\cN[\sigma_{AB}]} \underbrace{- \delta(q_{\top}) \log_2 d_Z - (1+\delta(q_{\top})) h_2 \left(\frac{\delta(q_{\top})}{1 + \delta(q_{\top})}\right)}_{-v(\delta)} \nonumber\\
    &\geq \tilde{g}(q) - g_{\corr}^{(\nu_c)}(q_{\top}).
\end{align}
Here, we denote $\cS_q$ as a shorthand for the feasible set defined in Eq.~\eqref{eq: annex original optimisation}. The second line is due to Lemma~\ref{lemma: continuity entropy}. The last line can be deduced by identifying that the constraints $\Delta(\rho_{AB}, \sigma_{AB}) \leq \delta(q_{\top})$ and $\rho_{AB} \in \cS_q$ are equivalent to the one written in Eq.~\eqref{eq: annex truncated optimisation}. Additionally, we notice that the correction terms $v(\delta(q_\top))$ is upper bounded by $g_{\corr}^{(\nu_c)}(q_\top)$, as derived previously. This concludes the proof of Theorem~\ref{thm: dimension reduction}.

\section{Numerical analysis} \label{annex: numerical analysis}

\subsection{Channel Model} \label{annex: model}

For numerical simulation, we consider a channel with channel loss $\eta$ along with excess noise $\chi$. After the prepared state $\ket{\psi_{X_j}} = \ket{\alpha e^{i (\pi X_j/2 + \pi/4)}}$ passes through the channel, the probability that it would result in a heterodyne detection in some phase space region $\mathcal{R}$ with $\abs{Y_j}^2\in[\tau_1^2,\tau_2^2]$ and $\theta\in[z\pi/2,(z+1)\pi/2]$ is
\begin{equation} \label{eq: XZ_distribution}
    \Pr[Y_j\in\mathcal{R}]=\frac{1}{\pi}\int_{\tau_1}^{\tau_2}d\beta\int_{\frac{z\pi}{2}}^{\frac{(z+1)\pi}{2}}d\theta \frac{\beta}{1+\frac{\eta\zeta}{2}} e^{\frac{-\beta^2-\eta\abs{\alpha}^2 + 2\abs{\alpha}\beta\sqrt{\eta}\cos(\theta-\frac{(2X_j+1)\pi}{4})}{1+\frac{\eta\zeta}{2}}}.
\end{equation}
We can compute the integral analytically for $\abs{Y_j}$ from 0 to infinity when considering regions extending to $\infty$,
\begin{equation}
    \Pr[\theta\in[z\pi/2,(z+1)\pi/2]]=\frac{1}{2\pi} \int_{\frac{z\pi}{2}}^{\frac{(z+1)\pi}{2}}d\theta e^{-\varphi^2} \left\{1+\sqrt{\pi}\varphi f(\theta)e^{\varphi^2f(\theta)^2}\text{erf}\left[\varphi f(\theta)+ 1\right]\right\},
\end{equation}
where $\varphi=\sqrt{\frac{\eta\abs{\alpha}^2}{1+\frac{\eta\zeta}{2}}}$ and $f(\theta)=\cos(\theta-\frac{(2X_j+1)\pi}{4})$.\\

\subsection{Min-tradeoff Function} \label{annex: min-tradeoff}

As discussed in Sec.~\ref{sec: SDP to bound entropy}, one could compute a lower bound on the conditional von Neumann entropy by solving the SDP in Eqn.~\eqref{eq: final SDP}.
We first define the score $C$ (excluding $\perp$) determined by Alice's state preparation choice and the measurement outcome discretisation for $T_j=1$,
\begin{equation}
    C_j=\begin{cases}
        \top & Z_j=\top\\
        0 & \cap_{k\in[0,3]}\{X_j=k\land Z_j=k\,(\text{mod}\,4)\}\\
        1 & \cap_{k\in[0,3]}\{X_j=k\land Z_j=k+1\,(\text{mod}\,4)\}\\
        2 & \cap_{k\in[0,3]}\{X_j=k\land Z_j=k+2\,(\text{mod}\,4)\}\\
        3 & \cap_{k\in[0,3]}\{X_j=k\land Z_j=k+3\,(\text{mod}\,4)\}\\
        (\varnothing,0) & \cap_{k\in[0,3]}\{X_j=k\land Z_j=(\varnothing,k)\,(\text{mod}\,4)\}\\
        (\varnothing,1) & \cap_{k\in[0,3]}\{X_j=k\land Z_j=(\varnothing,k+1\,(\text{mod}\,4))\}\\
        (\varnothing,2) & \cap_{k\in[0,3]}\{X_j=k\land Z_j=(\varnothing,k+2\,(\text{mod}\,4))\}\\
        (\varnothing,3) & \cap_{k\in[0,3]}\{X_j=k\land Z_j=(\varnothing,k+3\,(\text{mod}\,4))\}
    \end{cases}
\end{equation}
Based on the protocol, we can model the projectors of the measurement operators, $\tilde{\Pi}_c$, when they are projected into the photon number subspace for $n\leq n_{\max}$.
The POVM $\Pi_c$ projects the state in phase space to a region $\mathcal{R}_c^a$ (conditioned on $A=a$),
\begin{equation}
    \Pi_c=\sum_a\ketbra{a}{a}\otimes\frac{1}{\pi}\int_{\mathcal{R}^a_c}dY_j \ketbra{Y_j}{Y_j}.
\end{equation}
The projected POVM is simply
\begin{equation}
\begin{split}
    \tilde{\Pi}_c=&\sum_a\ketbra{a}{a}\otimes\sum_{n,n'\leq n_{\max}}\frac{1}{\pi}\int_{\mathcal{R}^a_c}dY_j \braket{n}{Y_j}\braket{Y_j}{n'}\ket{n}\bra{n'}\\
    =&\sum_a\ketbra{a}{a}\otimes\sum_{n,n'\leq n_{\max}}\frac{1}{\pi}\int_{\mathcal{R}^a_c}dY_j e^{-\abs{Y_j}^2}\frac{Y_j^n(Y_j^*)^{n'}}{\sqrt{n!n'!}}\ket{n}\bra{n'},
\end{split}
\end{equation}
where the second line notes $\braket{n}{Y_j}=e^{-\frac{\abs{Y_j}^2}{2}}\frac{Y_j^n}{\sqrt{n!}}$.
We can compute the overlap
\begin{equation}
    \Tr_Q[\ket{\Psi}\bra{\Psi}_{AQ}]=\frac{1}{4}\sum_{X_jX_j'}e^{-\abs{\alpha}^2(1-i^{X_j-X_j'})}\ket{X_j}\bra{X_j'}.
\end{equation}
Having these terms allows us to solve the SDP numerically with mosek~\cite{MOSEK} in Matlab and obtain the dual solution.\\

Solving the SDP with any trial parameter $q_{\text{dual}}$ with unit length provides a dual 
\begin{equation}
    \tilde{g}^{(q_{\text{dual}})}(q_{\text{dual}})=\varphi(\vbf{\lambda}_{q_{\text{dual}}})+\vbf{\lambda}_{q_{\text{dual}}}\cdot \vbf{h}(q_{\text{dual}}),
\end{equation}
where $\tilde{g}^{(q_{\text{dual}})}$ explicitly indicates that the dual function is obtained, and $\vbf{h}$ is some fixed function independent of the dual variables.
We note here that $\vbf{\lambda}_{q_{\text{dual}}}$ refer to the optimal dual variable when the SDP is solved with $q_{\text{dual}}$, i.e.
\begin{equation}
    \vbf{\lambda}_{q_{\text{dual}}}=\text{argmax}_{\vbf{\lambda}}\varphi(\vbf{\lambda})+\vbf{\lambda}\cdot \vbf{h}(q_{\text{dual}}).
\end{equation}
Importantly, for any $q$, we can shown that $\tilde{g}^{(q_{\text{dual}})}(q)$ lower bounds $\tilde{g}^{(q)}(q)$,
\begin{equation}
\begin{split}
    \tilde{g}^{(q)}(q)=&\text{max}_{\vbf{\lambda}}\varphi(\vbf{\lambda})+\vbf{\lambda}\cdot \vbf{h}(q)\\
    \geq &\varphi(\vbf{\lambda}_{q_{\text{dual}}})+\vbf{\lambda}_{q_{\text{dual}}}\cdot \vbf{h}(q)\\
    =&\tilde{g}^{(q_{\text{dual}})}(q).
\end{split}
\end{equation}
As such, we can always select one value of $q_{\text{dual}}$, and utilise $\tilde{g}^{(q_{\text{dual}})}(q)$ to construct the min-tradeoff function with the same linearisation method via Eq.~\eqref{eq: linearisation method} and conversion via Eq.~\eqref{eq: min-tradoff conversion}.
Since this value always lower bounds the conditional von Neumann entropy for all $q$, the resulting function is a valid min-tradeoff function.
In the numerical optimisation, the choice of $q_{\text{dual}}$ is optimised.
For simplicity, we optimise only over $q_{\text{dual}}$ generated from the same probability distribution as the model, with the same channel loss $\eta$ and amplitude $\alpha$.
The only parameter that could differ from the actual channel is $\chi_{\text{dual}}$, which we optimise over.

\end{document}